\documentclass[11pt]{article}
\usepackage{amsmath,amssymb,amsfonts,amsthm,epsfig,verbatim}
\usepackage{times}
\usepackage{color}
\usepackage{framed}

\newif\ifhyper\IfFileExists{hyperref.sty}{\hypertrue}{\hyperfalse}
\hypertrue
\ifhyper\usepackage{hyperref}\fi
\newenvironment{proofof}[1]{\par{\noindent \bf Proof of #1:}}{\qed\par}

\newtheorem{theorem}{Theorem}

\newtheorem{lemma}[theorem]{Lemma}
\newtheorem{proposition}[theorem]{Proposition}
\newtheorem{corollary}[theorem]{Corollary}
\newtheorem{claim}[theorem]{Claim}
\newtheorem{fact}[theorem]{Fact}

\newtheorem{definition}[theorem]{Definition}
\newtheorem{remark}[theorem]{Remark}
\newtheorem{question}[theorem]{Question}

\newcommand{\eps}{\epsilon}

\newtheorem*{theorem*}{Theorem}

\newcommand{\sgn}{\mathrm{sign}}

\renewcommand{\span}{\mathrm{span}}

\newcommand{\ignore}[1]{}

\newcommand{\cref}[1]{Corollary~\ref{cor:#1}}

\newcommand{\inp}[2]{\langle #1, #2\rangle}

\newcommand{\bits}{\{-1,1\}}
\newcommand{\bn}{\bits^n}

\newcommand{\R}{{\mathbb{R}}}

\newcommand{\E}{\operatorname{{\bf E}}}

\newcommand{\wt}[1]{\widetilde{#1}}

\newcommand{\poly}{\mathrm{poly}}

\newcommand{\wh}{\widehat}

\renewcommand{\Pr}{\operatorname{{\bf Pr}}}
\newcommand{\dist}{\mathrm{dist}}

\date{}

\begin{document}

\setcounter{page}{0}

\title{Is your function low-dimensional?}

\author{
Anindya De\footnote{Northwestern. Email: \texttt{ anindya@eecs.northwestern.edu} Supported by NSF grant CCF 1814706} \and
Elchanan Mossel \footnote{MIT. Email: \texttt{elmos@mit.edu}. Partially supported by NSF award DMS-1737944 and
ONR award N00014-17-1-2598} \and
Joe Neeman\footnote{UT Austin. Email: \texttt{jneeman@math.utexas.edu}.}
}

\maketitle

\thispagestyle{empty}

\begin{abstract}

We study the problem of testing if a function depends on a small number of linear directions of its input data. We call a function $f$  a \emph{linear $k$-junta} if it is completely determined by some $k$-dimensional subspace of the input space. In this paper, we study the problem of testing whether a given $n$ variable function $f : \R^n \to \{0,1\}$, is a linear $k$-junta or $\epsilon$-far from all linear $k$-juntas, where the closeness is measured with respect to the Gaussian measure on $\R^n$. Linear $k$-juntas are a common generalization of two fundamental classes from Boolean function analysis (both of which have been studied in property testing) \textbf{1.} $k$- juntas which are functions on the Boolean
cube which depend on at most k of the variables and \textbf{2.}  intersection of $k$ halfspaces, a fundamental geometric concept class.  

We show that the class of linear $k$-juntas is not testable, but adding a surface area constraint makes it testable: we give a $\mathsf{poly}(k \cdot s/\epsilon)$-query non-adaptive tester for linear $k$-juntas with surface area at most $s$. We show that the polynomial dependence on $s$ is necessary.
Moreover, we show that if the function is a linear $k$-junta with surface area at most $s$, 
we give a $(s \cdot k)^{O(k)}$-query non-adaptive algorithm to learn the function \emph{up to a rotation of the basis}.  In particular, this implies that we can test the class of intersections of $k$ halfspaces in $\R^n$ with query complexity independent of $n$.

\end{abstract}
\newpage

\section{Introduction}
 Property testing of Boolean functions was  initiated in the seminal work of Blum, Luby and Rubinfeld~\cite{BLR93} and Rubinfeld and Sudan~\cite{RS96}. The high level goal of property testing is the following: Given (query) access to a Boolean function $f$, the algorithm must distinguish between (i) the case that $f$ belongs to a class $\mathcal{C}$ of Boolean functions (i.e., \emph{has a property $\mathcal{C}$}), and (ii) the case that $f$ is $\epsilon$-far from every function belonging to $\mathcal{C}$. Here the distance between functions is measured with respect to some underlying distribution $\mathcal{D}$ and is defined as $\mathsf{dist}(f,g) = \Pr_{x \sim \mathcal{D}} [f(x) \not = g(x)]$. Also, the algorithm is randomized and thus only needs to succeed with high probability (as opposed to probability one). The quality of a testing algorithm is measured by the number of oracle calls it makes to $f$ -- its \emph{query complexity} -- and the goal is to minimize this query complexity.

Since the works of \cite{BLR93, RS96}, property testing of Boolean functions has been a thriving field and by now several classes $\mathcal{C}$ have been studied from this perspective. These include classes such as linear functions~\cite{BLR93}, low-degree polynomials~\cite{jutpatrudzuc04, bhattacharyya2010optimal}, monotonicity~\cite{FLNRRS, chakrabarty2016n, khot2015monotonicity}, algebraic properties~\cite{KaufmanSudan:08, bhattacharyya2015unified, bhattacharyya2013} and juntas~\cite{FKRSS03, blais2009testing, Chen:2017:SQC} among many others (see the surveys~\cite{ron2010algorithmic, goldreich_2017}).

Special attention has been devoted to the problem of testing juntas.  Recall
that a Boolean function $f: \{-1,1\}^n \rightarrow \{-1,1\}$ is said to be a
$k$-junta if $f$ is only dependent on a subset $S \subseteq [n]$ (of size  $k$)
of the coordinates. Given (query) access to a function $f$, the problem of
testing juntas is to decide whether $f$ is a $k$-junta or $\epsilon$-far from
every $k$-junta (under the uniform distribution on $\{-1,1\}^n$). Some of the
initial motivation~\cite{FKRSS03} to study this  came from the problem of long-code
testing~\cite{belgolsud98, PRS02} (related to PCPs and inapproximability).
Another motivation comes from the \emph{feature selection} problem in machine
learning. It is well-known (see, e.g.~\cite{Blum:94, BlumLangley:97}) that
learning a $k$-junta requires at least $\Omega(k \log n)$ samples, however
$k$-juntas can be tested with query complexity independent of
$n$~\cite{FKRSS03}.

The most obvious generalization of $k$-juntas to functions $f: \R^n \to \{-1, 1\}$ is
to consider functions that depend only on $k$ of the $n$ coordinates. However,
in many statistical and machine learning models (e.g. PCA, ICA, kernel learning,
dictionary learning) the choice of basis is not a priori clear.
Therefore, it is natural to consider a notion of junta that is linearly invariant.
We define a function $f: \mathbb{R}^n \rightarrow \{-1,1\}$ to be a {\em linear $k$-junta} if there are $k$ unit vectors $u_1, \ldots, u_k \in \mathbb{R}^n$ and $g: \mathbb{R}^k \rightarrow \{-1,1\}$ such that $f(x) = g(\langle u_1, x \rangle,  \ldots, \langle u_k, x \rangle)$.

We note that the family of linear $k$-juntas includes important classes of functions that have been studied in the learning and testing literature. Notably it includes: 
\begin{itemize}
\item Boolean juntas: If $h : \{-1,1\}^n \to \{0,1\}$ is a Boolean junta, then 
$f(x) : \R^n \to \{0,1\}$ defined as $f(x) = h(\mathsf{sgn}(x_1),\ldots,\mathsf{sgn}(x_n))$ is a linear $k$-junta.  
\item {Functions of halfspaces: Linear $k$-juntas include as a special case both halfspaces and intersections of $k$-halfspaces. The testability of halfspaces was studied in ~\cite{MORS:09,MORS:10,RS:13}.} 
\end{itemize} 
{We consider the scenario where the ambient dimension $n$ is large but the dimension of the relevant subspace, i.e., $k$ is small. In this setting, we consider the following property testing question:}
\begin{question} \label{q:1}
Given a function $f$ and access to random examples, $(x,f(x))$, is it possible to test in number of queries  that depends on $k$ (but not on $n$) if 
$f$ is  a linear $k$-junta or far from all linear $k$-juntas?
\end{question} 
The problem of testing linear-juntas is closely related to the problem of {\em model compression} in machine learning. 
The goal of model compression is to take as an input a complex predictor/classifier function and to output a simpler predictor/classifier  see e.g
~\cite{bucilua2006model}. 
The question of model compression is extensively studied in the context of deep nets, see e.g.,~\cite{ba2014deep},
and follow up work, 
where the models are often rotationally invariant (with the caveat that the regularization often used in optimization might not be). Thus as a motivating example we may ask: given a complex deep net classifier, is there a classifier that has essentially the same performance and depends only on $k$ of the features?

To formally state question~\ref{q:1} we need to define what ``close" means. The standard definition is to state that $f$ is close to $g$ if $\Pr[f(x) \neq g(x)]$ is small, for some probability measure $\Pr$. The most natural choice of $\Pr$ for learning and testing  functions 
$f : \R^n \to \{-1,1\}$ is the Gaussian measure~\cite{MORS:09, kothari2014testing, neeman2014testing, balcan2012active, chen2017sample, KOS:08, vempala2010learning, diakonikolas2018learning, balcan2013active, harsha2012invariance}. It is particularly natural in our setup since the Gaussian measure is invariant under many linear transformation, e.g., all rotations.  

It is possible to show that the answer to question~\ref{q:1} is {\bf no} even if $n=2$ and $k=1$, since without smoothness assumptions,   
measurable functions $f: \mathbb{R}^n \rightarrow \{-1,1\}$ can look arbitrarily random to any finite number of queries (a more formal statement with stronger results will be discussed shortly). 
Since the groundbreaking work of~\cite{KOS:08}, it was recognized that the surface area of a function $f: \mathbb{R}^n \rightarrow \{-1,1\}$ is a natural complexity parameter (see Definition~\ref{def:surface-area}  for the definition of surface area -- roughly speaking, if $A = \{x: f(x)=1\}$, then the surface area of $f$ is the size of the boundary of $A$ weighted by the Gaussian measure). 
We therefore ask the following question: 
\begin{question} 
Given a function $f$ and access to random examples, $(x,f(x))$, is it possible to test in number of queries that depends on $k$ and $s$ (but not on $n$) if $f$ is close to any linear $k$-junta with surface area at most $s$?
\end{question} 
In our main result we give an affirmative answer to the question above: 
\begin{theorem*}
There is an algorithm \textsf{Test-linear-junta} which has the following guarantee: Given oracle access to $f : \mathbb{R}^n \rightarrow \{-1,1\}$, rank parameter $k$, surface area parameter $s$ and error parameter $\epsilon>0$, it makes $\mathsf{poly}(s,\epsilon^{-1}, k)$ queries and distinguishes between the following cases: 
\begin{enumerate}
\item The function $f$ is a linear $k$-junta whose surface area is at most $s$. 
\item {The function $f$ is $\epsilon$-far from any linear $k$-junta with surface area at most $s(1+\epsilon)$.} 
\end{enumerate}
\end{theorem*}
This theorem is proven in Section~\ref{sec:test-rank}. 
We note that while the tester allows a slack of $1+\epsilon$ in 
the surface area between the soundness and completeness cases, such a slack factor is required even for the easier problem of estimating surface area in $\mathbb{R}^2$~\cite{neeman2014testing}.
 It is natural to ask if our dependence on the surface area is optimal. 
Towards answering this, in 
Section~\ref{sec:lb}, we prove: 
\begin{theorem*}
Any non-adaptive algorithm for testing whether an unknown Boolean function $f$ is a linear $1$-junta with surface area at most $s$ versus $\Omega(1)$-far from a linear $1$-junta makes at least $s^{\frac{1}{10}}$ queries. 
\end{theorem*} 
Thus our tester  is optimal in the dependence on $s$ up to polynomial factors. 

\subsubsection*{Finding the linear-invariant structure} 
Given the previous theorem it is natural to ask for more, i.e., not just test if the function is a linear-junta but also find the junta in number of queries that depends only on $k$ and $s$ (but not on $n$). 
 In other words, could we output $g: \mathbb{R}^k \rightarrow \{-1,1\}$ such that there exists a {projection} matrix $A: \mathbb{R}^n \rightarrow \mathbb{R}^k$ and 
 $f$ is close to $g(A x)$ with query complexity independent of $n$? We give an affirmative answer to this question: 
 \begin{theorem*}
 Let $f:\mathbb{R}^n \rightarrow \{-1,1\}$ be a linear $k$-junta with surface area at most $s$. Then, there is  an algorithm \textsf{Find-invariant-structure} which on error parameter $\epsilon>0$, makes $(s \cdot k/\epsilon)^{O(k)}$ queries and outputs $g: \mathbb{R}^k \rightarrow [-1,1]$ so that the following holds:  there exists an orthonormal set of vectors $w_1, \ldots, w_k \in \mathbb{R}^n$ such that 
 $$
 \mathbf{E}[|f(x)  - g(\langle w_1, x\rangle, \ldots, \langle w_k, x \rangle)|] = O(\epsilon). 
 $$
Moreover,  for some $g^{\ast} : \R^k \to \R$: 
 $$
 f(x) = g^{\ast}(\langle w_1, x\rangle, \ldots, \langle w_k, x \rangle).
 $$

 \end{theorem*}
Informally, the theorem states that it is possible to find the ``linear-invariant" structure (i.e., the structure up to unitary transformation) of $f$ in number of queries that dependens on $s$ and $k$. Of course, one cannot hope to output the relevant directions $w_1, \ldots, w_k$ explicitly as even describing these directions will require $\omega(n)$ bits of information and thus, at least those many queries. We note that the number of functions in $k$ dimensions with $O(1)$ surface area (even up to a unitary rotation) is $\exp (\exp (k))$ and thus even our output has to be $\exp(k)$ bits. Thus, it is not possible to significantly improve on our $\exp(k  \log k)$ query complexity in finding the linear-invariant structure.
  
 {
 \subsubsection*{Testability of linear invariant families of linear $k$-juntas} 
 
 Our ability to find the linear-invariant structure of linear $k$-juntas additionally allows us to test subclasses of linear $k$-juntas which are closed under rotation.
 
 \begin{definition}
 Let $\mathcal{C}$ be any collection of functions mapping $\mathbb{R}^k$ to $\{-1,1\}$. For any $n \in \mathbb{N}$ let:  
 \[
 \mathsf{Ind}(\mathcal{C})_n= \{f : \exists g \in \mathcal{C} \ \textrm{and orthonormal vectors } w_1, \ldots, w_k \textrm{ such that } f(x) = g(\langle w_1, x\rangle, \ldots, \langle w_k,x \rangle). \} 
 \]
Define $\mathsf{Ind}(\mathcal{C}) = \cup_{n=k}^\infty \mathsf{Ind}(\mathcal{C})_n$  and call it the \emph{induced class of $\mathcal{C}$}. 
\end{definition}
The two key properties of $\mathsf{Ind}(\mathcal{C})$ are (i) each function $f \in \mathsf{Ind}(\mathcal{C})$ is a linear $k$-junta, (ii) the class $\mathsf{Ind}(\mathcal{C})$ is closed under unitary 
 transformations. 
 The  definition is a continuous analogue of
 the so-called ``induced subclass of $k$-dimensional functions" from \cite{gopalan2009testing} (that paper was about testing functions over $\mathsf{GF}^n[2]$). 
 The following theorem shows that for any $\mathcal{C}$, $\mathsf{Ind}(\mathcal{C})$ is testable without any dependence on the ambient dimension. 

 \begin{theorem*}
 Let $\mathcal{C}$ be a collection of functions mapping $\mathbb{R}^k$ to $\{-1,1\}$. Further, for every $f \in \mathsf{Ind}(\mathcal{C})$, $\mathsf{surf}(f) \le s$. Then, there is an algorithm \textsf{Test-structure-$\mathcal{C}$} which has the following guarantee: Given oracle access to $f: \mathbb{R}^n \rightarrow \{-1,1\}$ and an error parameter $\epsilon>0$, the algorithm makes $(s \cdot k/\epsilon)^{O(k)}$ queries and distinguishes between the cases (i) $f \in \mathsf{Ind}(\mathcal{C})$ and (ii) $f$ is $\epsilon$-far from every function $g \in \mathsf{Ind}(\mathcal{C})$. 
 \end{theorem*} 
 

 A particularly important instantiation of the 
 above theorem is the following: Let $\mathcal{C}_{B}$ be any collection of functions mapping $\{-1,1\}^k \rightarrow 
 \{-1,1\}$ and let $\mathcal{C}$ be defined as 
 \[
 \mathcal{C} = \{g: x \mapsto h(\langle w_1, x \rangle - \theta_1, \ldots, \langle w_k, x \rangle - \theta_k) | \ w_1,\ldots, w_k \in \mathbb{R}^k, \ \theta_1, \ldots, \theta_k \in \mathbb{R}, \ h \in \mathcal{C}_B \}. 
 \]
 Note that $\mathcal{C}$ defined above is the set of functions obtained by composing a function from $\mathcal{C}_B$ with $k$-dimensional halfspaces. Consequently, $\mathsf{Ind}(\mathcal{C})$ is the of all functions which can be obtained by composing a function from $\mathcal{C}_B$ with halfspaces. As an example, if $\mathcal{C}_B$ consists of the $\mathsf{AND}$ function on $k$ or fewer bits, then $\mathsf{Ind}(\mathcal{C})$ is the class of ``intersections of $k$-halfspaces". Since the surface area of any Boolean function of $k$-halfspaces is bounded by $O(k)$ it follows that the this class is testable with $(k/\epsilon)^{O(k)}$ queries.

 Roughly speaking, the algorithm \textsf{Test-structure-$\mathcal{C}$} works as follows: 
we first run the routine \textsf{Test-linear-junta} -- if the target function $f$ passes this test, we are guaranteed that it is (very close to) a linear $k$-junta with surface area $s$. We then run the routine \textsf{Find-invariant-structure}. If the output of this step is $g$, then we can check whether $g$ is close to some function in $\mathsf{Ind}(\mathcal{C})_k$ and accept accordingly. We crucially note here that the last step, namely checking whether $g$ is close to a function in $\mathsf{Ind}(\mathcal{C})_k$ makes no queries to $f$.  While the overall intuition of this procedure is obvious, the precise  proof is more delicate and is given in Section~\ref{aff:inv}. 
}

\subsection{Related Work} 

\paragraph{Testing Boolean juntas}

As we have already mentioned, the problem of testing juntas on $\{-1, 1\}^n$
has already been well-studied. For example, it is known~\cite{blais2009testing,Chen:2017:SQC} that $\tilde \Theta(k^{3/2})$
queries are necessary and sufficient for non-adaptively testing $k$-juntas
with respect to the uniform distribution, while $\tilde \Theta(k)$ queries
are necessary and sufficient in the adaptive setting~\cite{blais2012property}.
It even turns out to be possible to test $k$-juntas with respect to an
unknown distribution~\cite{CLSSX18}, although in that setting the non-adaptive
query complexity becomes exponential in $k$.
{{We emphasize that while the problem of junta testing inspires the problems considered in this paper, junta testing algorithms have no bearing on the problem of testing linear juntas
-- e.g., unlike~\cite{CLSSX18}, there is  no  reason to believe that distribution-free testing of
linear juntas on $\R^n$ is even possible, given that the space of probability
measures on $\R^n$ is much richer than the space of probability measures on $\{-1, 1\}^n$.}}

\paragraph{Learning juntas of half-spaces.} 

{There has been extensive work on {\em learning} intersections and other functions of $k$  half-spaces~\cite{BlumKannan:97, vempala2010random, VX13, KOS:08} . 
Note that these algorithms (necessarily) require time polynomial in $n$ (whereas our \emph{raison d'etre is a query complexity independent of $n$}). In particular, 
\cite{BlumKannan:97} provided conditions under which intersections of halfspaces can be learnt under the uniform distribution on the ball. 
Vempala~\cite{vempala2010random} extended  their result to  arbitrary log-concave distributions. 
In terms of the expressivity of the function class, \cite{VX13} explicitly considered the problem of learning linear $k$-juntas (they called it subspace juntas) and showed that a linear $k$-junta of the form 
 $g(\langle w_1, x \rangle, \ldots, \langle w_k, x \rangle)$ is learnable in polynomial time if the function $g$ is identified by low moments and robust to small rotations in $\mathbb{R}^n$. Along a related but different axis, \cite{KOS:08} showed that functions of bounded surface area in the Gaussian space are learnable in polynomial time. Finally, we remark that there also has been work in learning intersections and other functions of halfspaces over the Boolean hypercube as well~\cite{KOS:02, gopalan2012learning}. 
}
 \ignore{learned some functions $g$ of $k$ half-spaces in polynomial time if the functions $g$ are identified by low moments and robust to small rotations in $\R^n$, while \cite{KOS:08} learned functions of bounded Gaussian surface area.}


\paragraph{Linearly Invariant Testing over Finite Fields}
We note that the set of linear-juntas is linearly invariant. If $f$ is a linear $k$-junta and $B$ is any $n \times n$ matrix then $x \mapsto f(Bx)$ is also a linear $k$-junta. 
Over finite fields, \cite{KaufmanSudan:08} studied general criteria for when a linearly invariant property is testable, see also 
\cite{bhattacharyya2013}. In particular, \cite{gopalan2009testing}, gave a $2^{O(k)}$ query complexity algorithm to test linear juntas over finite fields. Moreover, they also show that an exponential lower bound on $k$ is  necessary. 
This should be contrasted with our result which shows that linear juntas over the Gaussian space can be tested with $\mathsf{poly}(k)$ queries. 

{\paragraph{Testing (functions) of halfspaces}
The question of testing halfspaces was first considered in \cite{MORS:10} who showed that in the Gaussian space (as well as the Boolean space), halfspaces are testable with $O(1)$ queries. Subsequently, the second and third authors (Mossel and Neeman~\cite{mossel2015robust}) gave a different testing algorithm for a single halfspace in the Gaussian space. In fact, Harms~\cite{harms19} recently showed that halfspaces over any rotationally invariant distribution can be tested with sublinear number of queries. 
 However, as far as we are aware, prior to our work, no non-trivial bounds were known for even testing the intersection of two halfspaces. As remarked earlier, from our work, it follows that for any arbitrary $k$, intersection of $k$-halfspaces can be tested in the Gaussian space with $\exp(k \log  k)$ queries.}\\

\subsection{Techniques}
A major difference between linear juntas over finite fields and linear juntas over Gaussian space is the ``infinitesimal geometry" that can be used in the latter and does not exist in the former.
In particular, the linear part $\mathcal{W}_1(f)$ 
of the Hermite expansion of $f$ is approximately given by 
$e^{-t} (P_t f - \E[f])$ for large $t$. Here $P_t f$ is the Ornstein-Uhlenbeck operator. 
Both the quantities, $\E[f]$ and $P_t f$ can be approximated by sampling a small number of points from the Gaussian distribution and evaluating $f$ at those points. 
Moreover, if  $f(x) = g(\langle u_1, x \rangle,  \ldots, \langle u_k, x \rangle)$ is a linear junta, then the linear part of its Hermite expansion, $\mathcal{W}_1(f)$, lies in the span of $u_1,\ldots,u_k$. 

We would like to obtain  ``many more directions" that lie in the span of $u_1,\ldots,u_k$. 
We do so by considering functions of the form $f_{t,y}(x) = f(e^{-t} y + \sqrt{1-e^{-2 t}} x)$, for randomly chosen $y$ and an appropriate value of $t$ 
(the experts will recognize $f_{t,y}$ as part of the definition of the Ornstein-Uhlenbeck operator). Note that $f_{t,y}$ is also a linear junta defined by the same direction $u_1,\ldots,u_k$ and therefore the linear part of the Hermite expansion of $f_{t,y}$,  is also in the span of $u_1,\ldots,u_k$. 

It is now natural to propose the following algorithm to test if a function is a linear $k$-junta: choose points $y_i$ at random and ``compute''
$\mathcal{W}_1(f_{t,y_i})$ at these points. Then if the rank of the matrix spanned by 
$(\mathcal{W}_1(f_{t,y_i}))_i$ is at most $k$, then output YES; otherwise, output NO. 

Of course, actually computing $\mathcal{W}_1(f_{t,y})$ requires $\poly(n) \gg \poly(k)$ samples. 
Instead we will approximately compute the Gram matrix 
\[
A_{i,j} = \langle \mathcal{W}_1(f_{t,y_i}), \mathcal{W}_1(f_{t,y_j}) \rangle.
\]
and test if it is close or far from a matrix of rank $k$. 
One advantage of using the Gram matrix, is that we can evaluate the entries $A_{i,j}$ by sampling random inputs to evaluate the expected values
\[ 
\E[\mathcal{W}_1(f_{t,y_i})(x)  \mathcal{W}_1(f_{t,y_j})(x)].
\] 

How do we know that $\mathcal{W}_1(f_{t,y_i})(x)$ are not very close to $0$?
If $f$ has a bounded surface area then $f$ is close to the noise stable function $P_t f$. For such noise stable functions, we prove that with good probability at a random point $x$,
$\mathcal{W}_1(f_{t,y_i})(x)$ will be of non-negligible size. In fact, one of our main technical lemmas (Lemma \ref{lem:subspace-escape}) proves much more. It shows that if $f$ is $\epsilon$ far from any linear-$k$-junta then for any subspace $W$ with co-dimension at most $k$, it holds that for a random $y$ with probability at least $\poly(\epsilon)$, the projection of $\mathcal{W}_1(f_{t,y_i})(x)$ into $W$ will have norm at least 
$\poly(\epsilon)$. This result is later combined with a perturbation argument to establish to show that if $f$ is $\epsilon$-far from a linear $k$-junta then indeed the Gram matrix will have $k+1$ large eigenvalues. 
Since our analysis relies on the function $f$ having surface area at most $s$, the first stage of the algorithm uses the algorithm by the third author \cite{neeman2014testing} to test if the function of interest is of bounded surface area. 

The algorithm to identify  the linear invariant structure of $f$ builds up on the ideas in the algorithm to test linear $k$-juntas. More precisely, we can show that if $f$ is a linear $k$-junta with surface area $s$, 
\begin{enumerate}
\item we can find directions $y_1,\ldots, y_\ell$ such that 
$f$ is close to a function on the space spanned by the directions $\mathcal{W}_1(f_{t,y_1}),\ldots,\mathcal{W}_1(f_{t,y_\ell})$ (for some $\ell \le k$). 
\item While we cannot find $\mathcal{W}_1(f_{t,y_j})$ explicitly for any $j$, we can  evaluate $\langle \mathcal{W}_1(f_{t,y_j}), x\rangle$ at any point $x$ up to good accuracy. 
\item With the above observation, the high level idea is to \emph{try out all smooth functions} on the subspace spanned by $\{\langle\mathcal{W}_1(f_{t,y_1}), x\rangle ,\ldots,\langle \mathcal{W}_1(f_{t,y_\ell}),x \rangle\}$.  Perform \emph{hypothesis testing} for each such function against $f$ and output the most accurate one. 
\end{enumerate}
The crucial part in the above argument is that even if we have $\mathcal{W}_1(f_{t,y_1}),\ldots,\mathcal{W}_1(f_{t,y_\ell})$ implicitly, the space of ``all smooth functions" on $\mathsf{span}(\langle\mathcal{W}_1(f_{t,y_1}), x\rangle ,\ldots,\langle \mathcal{W}_1(f_{t,y_\ell}),x \rangle)$ has a cover whose size is independent of $n$. This lets us identify the linear invariant function defining $f$ with query complexity just dependent on $k$ and $s$.  

In order to prove lower bounds in terms of surface area, we construct a distribution over linear $1$-juntas
with large surface area by splitting $\R^2$ into many very thin parallel strips (oriented in a random direction)
and assign our function a random $\pm 1$ value on each strip. (Note that the surface area of such a function
is proportional to the number of strips.) The intuition is
that no algorithm that makes non-adaptive queries can tell that such a random
function is a 1-junta, because in order to ``see'' one of these strips, the
algorithm would need to have queried multiple far-away points in a single
strip. But if the number of queries is small relative to the number of strips
then this is impossible -- with high probability every pair of far-away query points
will end up in different strips.
In order to make this intuition rigorous, we also introduce a distribution on linear $2$-juntas
by randomly ``cutting'' the thin strips once in the orthogonal direction. We show that
for any non-adaptive set of queries, the two distributions induce almost identical query distributions,
and Yao's minimax lemma implies that no algorithm can distinguish between our random $1$-juntas and our
random $2$-juntas.

\vspace{-0.2cm}

\section{Preliminaries}

In this paper, unless explicitly mentioned otherwise, the domain $\mathbb{R}^n$ is always endowed with the measure $\gamma_n$, the standard $n$-dimensional Gaussian measure. Likewise, we will only consider functions $f \in L_2(\gamma_n)$. For such a function, and $t>0$, we recall that the so-called Ornstein-Uhlenbeck operator $P_t$ is defined as follows: 
\[
P_t f (x) = \int_y f(e^{-t} x + \sqrt{1-e^{-2t}}z) \gamma_n(z) dz
\]

We will also need to recall some very basic facts about Hermite expansion for functions $f \in L_2(\gamma_n)$. In particular, recall that for all $q \ge 0$, we can define the  Hermite polynomial $H_q: \mathbb{R} \rightarrow \mathbb{R}$ as 
\[
H_0(x) =1 ; \   H_q(x) = \frac{(-1)^q}{\sqrt{q!}}  \cdot e^{x^2/2}  \cdot \frac{d^q}{dx^q} e^{-x^2/2}. 
\]

Further, for the ambient space $\mathbb{R}^n$, let us
define the space $\mathcal{W}_q$ to be the linear subspace of $L_2(\gamma_n)$ spanned by $\{H_q(\langle v,x\rangle) : v \in \mathbb{S}_n\}$. Here $\mathbb{S}_n$ denotes the unit sphere in $n$-dimensions.  For a function $g \in L_2(\gamma_n)$, we let $\widehat{g}_q: \mathbb{R}^n \rightarrow \mathbb{R}$ denote the projection of $g$ to the subspace $\mathcal{W}_q$. Note that for any $g$, $\widehat{g}_q$ will be a degree-$q$ polynomial lying in the subspace $\mathcal{W}_q$.
We now recall some standard facts from Hermite analysis which can be found in any standard text on the subject (see
the book by O'Donnell~\cite{o2014analysis}. 
\begin{proposition}~\label{prop:Hermite-basics}
\begin{enumerate}
\item For $q \not =q'$, the subspaces $\mathcal{W}_q$ and $\mathcal{W}_{q'}$ are orthogonal. In other words, if $r \in \mathcal{W}_q$ and $s \in \mathcal{W}_{q'}$, then $\mathbf{E}_{x \sim \gamma_n}[r(x) \cdot s(x)]=0$.
\item Every function $g \in L_2(\gamma_n)$ can be expressed as $g (x) = \mathop{\sum}_{q \ge 0} \widehat{g}_q (x)$ where $\widehat{g}_q$ is the projection of $g$ to $\mathcal{W}_q$. 
\item For any $t>0$, $(P_t g)(x) = \sum_{q \ge 0} e^{-t \cdot q} \cdot \widehat{g}_q(x)$. 
\end{enumerate}
\end{proposition}

\subsubsection{Oracle computation}
We now list several useful claims which all fit the same motif: Given oracle access to $f: \mathbb{R}^n \rightarrow \mathbb{R}$, what \emph{interesting} quantities can be computed? 

\begin{lemma}~\label{lem:oracle-access-1}
Given oracle access to $f:\mathbb{R}^n \rightarrow [-1,1]$, error parameter $\eta>0$, there is a function $f_{\partial,\eta}: \mathbb{R}^{n} \rightarrow \mathbb{R}$ such that the following holds for every $\lambda \ge 1$, 
\[
\mathop{\Pr}_{x \sim \gamma_n} \big[ \big|f_{\partial,\eta}(x) - \widehat{f}_1(x) \big| > \lambda \cdot \eta \big] \le \lambda^{-2}. 
\]
Further, for any $x \in \mathbb{R}^n$, we can compute $f_{\partial,\eta}(x)$  to additive error $\pm \epsilon$ with confidence $1-\delta$ by making $\mathsf{poly}(1/\eta, 1/\epsilon, \log (1/\delta))$ queries to the oracle for $f$. 
\end{lemma}
\begin{proof}
Observe that for any $t>0$, $P_t f = \sum_{q \ge 0} e^{-tq} \widehat{f}_q(x)$. This implies that 
\[
\frac{P_t f - \mathbf{E}[f]}{e^{-t}} = \widehat{f}_1(x) + \sum_{q>1}e^{-t(q-1)} \widehat{f}_q(x). 
\]
Set $t$ so that $e^{-t} = \eta$ and let us define $f_{\partial,\eta}$ as
$
f_{\partial,\eta} = \frac{P_t f - \mathbf{E}[f]}{e^{-t}}. 
$ Now, observe that for $h(x)=\sum_{q>1}e^{-t(q-1)} \widehat{f}_q(x)$, $\mathbf{E}[h(x)]=0$ and $\mathsf{Var}[h(x)] \le \eta^2$. We now apply Chebyshev's inequality to obtain 
\[
\mathop{\Pr}_{x \sim \gamma_n} \big[ \big|f_{\partial,\eta}(x) - \widehat{f}_1(x) \big| > \lambda \cdot \eta \big] \le \lambda^{-2}. 
\]
Next, observe that both $P_tf(x)$ and $\mathbf{E}[f(x)]$ can be computed to error $\pm \epsilon \cdot \eta$ with confidence $1-\frac{\delta}{2}$ using $\mathsf{poly}(1/\eta, 1/\epsilon, \log (1/\delta))$ queries to the oracle for $f$. This immediately implies that $f_{\partial,\eta}$ can be computed to error $\pm \epsilon$ using $\mathsf{poly}(1/\eta, 1/\epsilon, \log (1/\delta))$ queries to the oracle for $f$. 
\end{proof}

\begin{lemma}~\label{lem:inner-products} 
    Given oracle access to functions $f,g : \mathbb{R}^n \rightarrow [-1,1]$, error parameter $\epsilon >0$ and confidence parameter $\delta>0$, there is an algorithm which makes $\mathsf{poly}(1/\epsilon,\log(1/\delta))$ queries to $f,g$ and computes $\langle \widehat{f}_1, \widehat{g}_1 \rangle$ up to error $\epsilon$ with confidence $1-\delta$. 
\end{lemma}
\begin{proof}
    Consider the function
    \[
        h(x) = e^{2t} (P_t f(x) - \E[f]) (P_t g(x) - \E[g]).
    \]
    Writing out the Fourier expansions of $P_t f$ and $P_t g$, note that $P_t f = \sum_{q \ge 0} e^{-tq} \widehat f_q(x)$,
    and so
    \[
        h(x) = \widehat f_1(x) \widehat g_1(x) + \sum_{\substack{q, r \ge 1 \\ q + r \ge 3}} e^{-t(q + r - 2)} \widehat f_q(x) \widehat g_r(x).
    \]
    Since $\widehat f_1$ and $\widehat g_1$ are linear functions, $\E[\widehat f_1(x) \widehat g_1(x)] = \widehat f_1 \cdot \widehat g_1$. On the other hand, $\E[\sum_{q \ge 0} \widehat f_q^2(x)] = \E[f^2] \le 1$, and so the Cauchy-Schwarz inequality implies that
    \[
        \E\Big[\sum_{\substack{q, r \ge 1 \\ q + r \ge 3}} e^{-t(q + r - 2)} \widehat f_q(x) \widehat g_r(x)\Big]
        \le e^{-t}.
    \]
    Hence, $|\E[h(x)] - \hat f_1 \cdot \hat g_1| \le e^{-t}$. If we choose $t$ so that $e^{-t} = \epsilon/2$, then it
    only remains to show that we can estimate $\E[h(x)]$ within additive error $\epsilon/2$ with confidence $1 - \delta$.

    Let $y$ and $z$ be Gaussian random variables, independent of $x$, and write $P_t f(x) = \E_y[f(e^{-t} x + \sqrt{1-e^{-2t}} y)]$ and $P_t g(x) = \E_z[g(e^{-t} x + \sqrt{1-e^{-2t}} z)]$. In particular, we can express $\E[h(x)]$ in the form $\E[J(x, y, z)]$ where
    \[
        J(x,y,z) = e^{2t} (f(e^{-t} x + \sqrt{1-e^{-2t}} y) - f(y)) (g(e^{-t} x + \sqrt{1-e^{-2t}}z) - g(z)).
    \]
    Recalling that $e^{-2t} = 4/\epsilon^2$, it follows that $J$ takes values in $[-4/\epsilon^2, 4/\epsilon^2]$,
    and it follows from Hoeffding's inequality that we can approximate $J$ to additive error $\epsilon/2$
    with confidence $1 - \delta$ using $\mathsf{poly}(1/\epsilon, \log(1/\delta))$ samples of $J$. Moreover,
    each sample of $J$ can be computed using two oracle queries to $f$ and two oracle queries to $g$. 
\end{proof}

\begin{definition}
A function $f: \mathbb{R}^n \rightarrow \mathbb{R}$ 
is said to be a linear $k$-junta if there are at most $k$ orthonormal vectors $u_1, \ldots, u_k \in \mathbb{R}^n$ and a function $g: \mathbb{R}^k \rightarrow \mathbb{R}$ such that 
\[
f(x) = g(\inp{u_1}{x}, \ldots, \inp{u_k}{x}).  
\]
Further, if $u_1, \ldots, u_k \in W$ (a linear subspace of $\mathbb{R}^n$), then $f$ is said to be a $W$-junta. 
%
%
\end{definition}

\subsection{Derivatives of functions}
{We will use $D$ to denote the derivative operator. In case, there are two sets of variables involved, we will explicitly indicate the variable with respect to which we are taking the derivative.} 

\begin{definition}
For $f: \mathbb{R}^n \rightarrow \mathbb{R}$ ($f \in \mathcal{C}^{\infty}$) and $t \ge 0$, define the function $f_t: \mathbb{R}^n \times \mathbb{R}^n \rightarrow \mathbb{R}$, 
\[
f_t(y,x) = f(e^{-t} y + \sqrt{1-e^{-2t}} x). 
\]
Further, in the same setting as above, we let 
$f_{t,y}: \mathbb{R}^n \rightarrow \mathbb{R}$, 
\[
f_{t,y}(x) = f(e^{-t} y + \sqrt{1-e^{-2t}} x). 
\]
\end{definition}

Let $D_x$ denote the derivative operator with respect to $x$ and let $D_y$ denote the derivative operator with respect to $y$. Then, it is easy to observe that 
\begin{equation}~\label{eq:derivative-y-1}
\sqrt{e^{2t}-1} \cdot D_y f_{t}(y,x)  = D_x f_{t}(y,x). 
\end{equation}
Next, for a function $g: \mathbb{R}^n \rightarrow \mathbb{R}$, define $\mathcal{W}_{1}(g) \in \mathbb{R}^n$ as the degree-$1$ Hermite coefficients of $g$. In other words, the $i^{th}$ coordinate of $\mathcal{W}_1(g)$ 
\[
\mathcal{W}_{1}(g)[i] = \mathbf{E}[g(x) \cdot x_i],
\]
where $x \sim \gamma_n$, the standard $n$-dimensional Gaussian measure. With respect to our earlier definition of $\widehat{g}_1$,  observe that we have: 
$
\widehat{g}_1(x) = \inp{\mathcal{W}_{1}(g)}{x}. 
$
We next prove the following important lemma which connects the gradient of $P_t f$ at $y$ with $\mathcal{W}_{1}(f_{t,y})$. In particular, we have the following lemma. 
\begin{lemma}~\label{lem:derivative-shift}
\[
\mathcal{W}_{1}(f_{t,y})= \sqrt{e^{2t}-1} \cdot D (P_t f)(y). 
\]
\end{lemma}
\begin{proof}
First of all, observe that for any function $g: \mathbb{R}^n \rightarrow \mathbb{R}$ with bounded derivatives, and for any $i \in [n]$
\begin{eqnarray*}
\mathbf{E}_{x \sim \gamma_n} \bigg[ \frac{\partial g(x)}{\partial x_i} \bigg] = \int_{x} \frac{\partial g(x)}{\partial x_i}\gamma_n(x) dx = \int_{x} x_i g(x) \gamma_n(x) dx = \mathbf{E}_{x \sim \gamma_n} [x_i \cdot g(x)]. 
\end{eqnarray*}
While the first and last equalities are trivial, the middle is a consequence of integration by parts. Assuming that
$f$ has bounded derivatives, we may apply this identity to $g = f_{t,y}$, yielding
\begin{eqnarray*}
\mathcal{W}_1(f_{t,y}) &=&  \mathbf{E}_x [D_xf_{t}(y,x)] 
\\ &=& \sqrt{e^{2t}-1} \cdot \mathbf{E}_x [D_y f_{t}(y,x)] \ \ \textrm{(applying (~\ref{eq:derivative-y-1}))}
 \\
&=& \sqrt{e^{2t}-1} \cdot D_y (\mathbf{E}_x [f_{t}(y,x)]) = \sqrt{e^{2t}-1} \cdot D_y (P_t f)(y). 
\end{eqnarray*}
This proves the lemma in the case that $f$ has bounded derivatives. In the general case, we approximate
choose a sequence of functions that have bounded derivatives and approximate $f_{t,y}$ in $L_2(\gamma)$.
Applying the lemma to these functions and taking the limit proves the general case.
\end{proof} 
\begin{lemma}~\label{lem:inner-product-1}
Given oracle access to $f$, noise parameter $t>0$, error parameter $\epsilon>0$, confidence parameter $\delta>0$ and $y_1, y_2 \in \mathbb{R}^n$, there is an algorithm which makes $\mathsf{poly}(1/\epsilon, 1/\delta, 1/t)$ queries to $f$ and computes $\langle D(P_t f)(y_1),D(P_t f)(y_2)\rangle$ up to error  $\epsilon$ with confidence $1-\delta$. 
\end{lemma}
\begin{proof}
By Lemma~\ref{lem:derivative-shift}, we have
\[
\langle D(P_t f)(y_1),D(P_t f)(y_2)\rangle = \frac{1}{e^{2t}-1} \cdot \langle \mathcal{W}_1(f_{t,y_1}), \mathcal{W}_1(f_{t,y_2})\rangle. 
\]
We can now apply Lemma~\ref{lem:inner-products} to finish the proof. 

\end{proof}
\begin{proposition}~\label{prop:derivative-bound}
For any $f: \mathbb{R}^n \rightarrow [-1,1]$, $\Vert D(P_t f)(y) \Vert_2 \le (e^{2t} -1)^{-\frac12}$.  
\end{proposition}
\begin{proof}
By Lemma~\ref{lem:derivative-shift}, we have $\Vert \mathcal{W}_1(f_{t,y}) \Vert_2 = \sqrt{e^{2t}-1} \cdot \Vert D(P_t f)(y) \Vert_2$. Now, observe that the range of $f_{t,y}$ is $[-1,1]$ and thus, $\Vert \mathcal{W}_1(f_{t,y}) \Vert_2\le 1$, implying the stated upper bound. 
\end{proof}

\begin{lemma}~\label{lem:compute-derivative-x}
Given oracle access to $f: \mathbb{R}^n \rightarrow [-1,1]$, $y \in \mathbb{R}^n$, noise parameter $t>0$,
error parameter $\eta>0$, there is a function 
$f_{\partial,\eta,t,y}: \mathbb{R}^n \rightarrow \mathbb{R}$ such that the following holds for every $\lambda \ge 1$,
\[
\Pr_{x \sim \gamma_n} [|f_{\partial,\eta,t,y}(x) - \langle D(P_t f)(y),x\rangle |  > \lambda \cdot \eta] \le \lambda^{-2}. 
\]
Further, for an error parameter $\epsilon>0$, confidence parameter $\delta>0$, we can compute 
$f_{ \partial, t,\eta,y}$ to additive error $\pm \epsilon$ with confidence $1-\delta$ using $\mathsf{poly}(1/t, 1/\eta, 1/\epsilon, \log(1/\delta))$ queries to $f$. 
\end{lemma}
\begin{proof}
We first use Lemma~\ref{lem:derivative-shift} and obtain that 
$$
D_{}P_tf(y) = \frac{1}{\sqrt{e^{2t}-1}} \cdot \mathcal{W}_1(f_{t,y}). 
$$
Consequently, we have that 
\[
\langle DP_tf(y), x\rangle = \frac{1}{\sqrt{e^{2t}-1}} \cdot \widehat{f_{t,y}}_1(x).
\]
The claim now follows from Lemma~\ref{lem:oracle-access-1}. 
\end{proof}

\subsection{Some useful inequalities  concerning noise stability}
\begin{lemma}~\label{lem:Poincare} \textbf{[Poincar\'{e} inequality]}  Let $f: \mathbb{R}^n \rightarrow \mathbb{R}$ be a $\mathcal{C}^1$ function. Then, $\mathsf{Var}[f] \le \mathbf{E}[\Vert Df \Vert_2^2]$. 
\end{lemma}


\begin{definition}~\label{def:surface-area}
For a Borel set $A \subseteq \mathbb{R}^n$, we define its Gaussian surface area $\Gamma(A)$ to be
\[
\Gamma(A) = \liminf_{\delta \rightarrow 0} \frac{\mathsf{vol}(A_{\delta} \setminus A)}{\delta},
\]
provided the limit exists. Here, for any body $K$, $\mathsf{vol}(K)$ denotes the Gaussian volume of $K$, i.e., $\int_{x \in K} \gamma_n(x) dx$. Further, $A_{\delta} = \{x : d(x,A) \le \delta\}$ where $d(x,A)$ denotes the Euclidean distance of $x$ from $A$. 

For a function $f : \mathbb{R}^n \rightarrow \{-1,1\}$, we denote its surface area $\Gamma(f) = \Gamma(A_f)$ where $A_f = \{x: f(x)=1\}$. 
\end{definition}

Ledoux~\cite{Ledoux:94} (and implicitly Pisier~\cite{Pisier:86}) proved the following connection between noise sensitivity and surface area of functions. 
\begin{lemma}~\label{lem:Ledoux}[Ledoux~\cite{Ledoux:94}]
For any $t \ge 0$ and $f : \mathbb{R}^n \rightarrow \{-1,1\}$, $x,y \sim \gamma_n$, we have \[
\Pr_{x,y} [f(x) \not = f(e^{-t} x+ \sqrt{1-e^{-2t}} y)] \le \frac{2\sqrt{t}}{\sqrt{\pi}} \cdot \Gamma(f)\] 
\end{lemma}
The following proposition is an immediate consequence of the above lemma. 
\begin{proposition}~\label{prop:Ledoux}
Let $f: \mathbb{R}^n \rightarrow \{-1,1\}$, $t \ge 0$ and $\Gamma(f) \le s$. Then, 
\begin{enumerate}
\item $\mathbf{E}[(f(x) - P_tf (x))^2] =  8 s \sqrt{t}$. 
\item For any $\epsilon >0$ and $T = O(s^2/\epsilon^2)$, $\sum_{q \ge T} \mathbf{E}[\widehat{f}_q^2] \le \epsilon$. 
\end{enumerate}
\end{proposition}
\begin{proof}
Let $\mathcal{E}_1(x,y)$ denote the event that 
$f(x) \not = f(e^{-t} x+ \sqrt{1-e^{-2t}} y)$. 
To prove the first item, observe that for any $x$, 
\[
(f(x) - P_t f(x))^2  = (2\mathop{\mathbf{E}}_{y \sim \gamma_n}[\mathbf{1}(\mathcal{E}_1(x,y))])^2 =4 \big(\mathop{\mathbf{E}}_{y \sim \gamma_n}[\mathbf{1}(\mathcal{E}_1(x,y))]\big)^2 \le 4 \big(\mathop{\mathbf{E}}_{y \sim \gamma_n}[\mathbf{1}(\mathcal{E}_1(x,y))]\big)
\]
Thus, we obtain that $$\mathbf{E}[(f(x) - P_tf(x))^2] \le 4\mathop{\mathbf{E}}_{x, y \sim \gamma_n}[\mathbf{1}(\mathcal{E}_1(x,y))] \le 8 s\sqrt{t}, $$
where the last inequality is an application of Lemma~\ref{lem:Ledoux}. The second item here is the same as Theorem~15 (full version) of ~\cite{KOS:08}. So, we do not prove it here. 
\end{proof}

\subsection{Inequalities for matrix perturbation}
We will require some basic results on matrix perturbations. For this, we adopt the following notation: Let $A \in \mathbb{C}^{n \times n}$ be a Hermitian matrix. Then $\sigma_1(A) \ge \ldots \ge \sigma_n(A)$ denote its singular values in order.

\begin{lemma}~\label{lem:Weyl}[Weyl's inequality]
Let $A, E \in \mathbb{R}^{n \times n}$ be real symmetric matrices. 
Then for any $j$, $$|\sigma_j(A+E) - \sigma_j(A)| \le \Vert E \Vert_F. $$
\end{lemma}

\begin{fact}~\label{lem:mat-1-bound}~\cite{schmitt1992perturbation}
Let $A_1, A_2$ be two psd matrices. Let $\sigma_{\min}(A_1), \sigma_{\min}(A_2) \ge c$. Then, 
\[
\Vert A_2^{1/2} - A_1^{1/2} \Vert_2 \le \Vert A_2 - A_1 \Vert_2 \cdot \frac{1}{2 \sqrt{c}}. 
\]
\end{fact} 
\begin{fact}~\label{lem:mat-2-bound}\cite{stewart1973introduction}
Let $A_1, A_2$ be two psd matrices. Let $\sigma_{\min}(A_1)\ge c$ and $\Vert A_2 - A_1 \Vert_2 \le c/100$. Then, 
\[
\Vert A_2^{-1} - A_1^{-1} \Vert_2 \le \Vert A_2 - A_1 \Vert_2 \cdot \frac{1}{ c^2}. 
\]
\end{fact} 
Combining these two facts, we have the following corollary.
\begin{corollary}~\label{corr:mat-perturb}
Let $0<c<1$ and let $A_1$ be a psd matrix such that $\sigma_{\min}(A_1) \ge c$. Let $A_2 - A_1$ be real symmetric that $\Vert A_2 - A_1 \Vert_2 \le \xi \cdot c$ for $|\xi| \le 1/100$. Then, 
$\Vert A_{1}^{-1/2} - A_2^{-1/2} \Vert_2 \le \frac{\xi}{2\sqrt{c}}$. 
\end{corollary}
\begin{proof}
We first apply Fact~\ref{lem:mat-1-bound} to obtain that $$
\Vert A_2^{1/2} - A_1^{1/2} \Vert_2 \le \frac{\xi c^{1/2}}{2}. 
$$
Observe that $\sigma_{\min}(A_1^{1/2}) \ge \sqrt{c}$. Since 
$c<1$ and $|\xi| \le \frac{1}{100}$,  we apply 
Fact~\ref{lem:mat-2-bound} to obtain that
$$
\Vert A_2^{-1/2} - A_1^{-1/2} \Vert_2 \le \frac{\xi}{2\sqrt{c}}. 
$$
This finishes the proof. 
\end{proof}


\section{Algorithm to test $k$-juntas}~\label{sec:test-rank}
In this section, we will prove the following theorem. 
\begin{theorem}~\label{thm:main1}
There is an algorithm \textsf{Test-linear-junta} which has the following guarantee: Given oracle access to $f: \mathbb{R}^n \rightarrow \{-1,1\}$, rank parameter $k$, surface area parameter $s$ and error parameter $\epsilon>0$, it makes $\mathsf{poly}(s,\epsilon^{-1}, k)$ queries and 
\begin{enumerate}
\item If $f$ is a linear $k$-junta with $\mathsf{surf}(f) \le s$, then the algorithm outputs \textsf{yes} with probability at least $0.9$. 
\item If $f$ is $O(\epsilon)$-far from any linear $k$-junta $g$ with {$\mathsf{surf}(g) \leq (1+\epsilon) \cdot s$}, then the algorithm outputs \textsf{no} with probability at least $0.9$. 
\end{enumerate}
\end{theorem} 
\begin{remark}{A convention that we shall adopt (to avoid proliferation of parameters) is to sometimes ignore the confidence parameter of the testing algorithm. Typically, whenever we can estimate a parameter within $\pm \epsilon$ with $T$ queries with confidence $2/3$, we can do the usual ``median trick" and get the same accuracy with confidence $1-\delta$ with a multiplicative $O(\log(1/\delta))$ overhead in the query complexity. Since we only need to succeed with probability $0.9$ in the final algorithm, it is sufficient for each of the individual subroutines to succeed with probability sufficiently close to $1$. So, unless it is crucial, at some places,we shall ignore the confidence parameter in the theorem statements and many of the calculations. It will be implicit that the confidence parameter is sufficiently close to $1$. 
 }
 \end{remark}
The algorithm \textsf{Test-linear-junta} is described in Figure~\ref{fig:tlj}. The algorithm invokes two different subroutines, \textsf{Test-surface-area} and \textsf{Test-rank} whose guarantees we state now. To
do this, we first define the notion of $(\epsilon, s)$ smooth function. 
\begin{definition}~\label{def:smooth-perturb}
A function $f: \mathbb{R}^n \rightarrow \{-1,1\}$ is said to be $(\epsilon, s)$-smooth if there is a function $g: \mathbb{R}^n \rightarrow \{-1,1\}$ such that $\mathbf{E}[|f-g|]\le\epsilon$ and $\mathsf{surf}(g) \le s (1+\epsilon)$. 
\end{definition}
In other words, a function $f$ is $(\epsilon,s)$ smooth if $f$ is $\epsilon$-close to some other function $g$ (in $\ell_1$ distance) and $g$ has surface area which is essentially bounded by $s$. With this definition, we can now state the guarantee of the routine \textsf{Test-surface-area} 
(due to Neeman~\cite{neeman2014testing}). 
\begin{theorem}~\label{thm:neeman-testing}
There is an algorithm \textsf{Test-surface-area} which given oracle access to a function $f: \mathbb{R}^n \rightarrow \{-1,1\}$ and error parameter $\epsilon>0$ makes $T_{\mathsf{test}} = \mathsf{poly}(s/\epsilon)$ queries and has the following guarantee: 
\begin{enumerate}
\item If $f$ is a function with surface area at most $s$, then the algorithm outputs \textsf{yes} with probability at least $1-\epsilon$. 
\item {Any function $f$ which passes the test with probability $0.1$ is $(\epsilon,s)$-smooth.} 
\end{enumerate}
\end{theorem}
Next, we state the guarantee of the routine \textsf{Test-rank}. 
\begin{lemma}~\label{lem:far-k-junta}
The routine \textsf{Test-rank} has a query complexity of $\mathsf{poly}(k,s, \epsilon^{-1})$. Further, we have
\begin{enumerate}
\item If the function $f$ is a linear-$k$-junta, then the algorithm \textsf{Test-rank} outputs \textsf{yes} with probability $1-\epsilon$. 
\item 
If $f : \mathbb{R}^n \rightarrow \{-1,1\}$  is a $((\epsilon/30)^2,s)$-smooth function which is $\epsilon$-far from a linear $k$-junta,
 then the algorithm \textsf{Test-rank} outputs \textsf{no} with probability $1-\epsilon$. 
 \end{enumerate}
\end{lemma}
{In order to prove Theorem~\ref{thm:main1}, we will need the following claim which shows that property of closeness to a linear $k$-junta and closeness to a smooth function can be certified using a single function.
\begin{lemma}~\label{lem:dual-closeness}
For a function $f: \R^n \to \{-1, 1\}$, suppose that there is a linear $k$-junta $g: \R^n \to \{-1, 1\}$
    and a function $h: \R^n \to \{-1, 1\}$ of surface area at most $s$ such that both $g$ and $h$
    are $\epsilon$-close to $f$. Then there is a function $\tilde h: \R^n \to \{-1, 1\}$ that is a linear $k$-junta \emph{and}
    has surface area at most $s(1 + \sqrt \epsilon)$, and which is $O(\sqrt{\epsilon})$-close to $f$.
\end{lemma}
}
{\begin{proofof}{Theorem~\ref{thm:main1}}
If $f$ is a linear $k$-junta with surface area at most $s$, then it passes both the tests \textsf{Test-surface-area} as well as \textsf{Test-rank} with probability $1-\epsilon$. Thus, any linear $k$-junta with surface area at most $s$ passes with probability at least $1-2\epsilon$ (so as long as $\epsilon \le 0.05$, the test succeeds with probability $0.9$). 

On the other hand, suppose $f$ passes \textsf{Test-linear-junta} with probability $0.9$. Then, applying Theorem~\ref{thm:neeman-testing} is $((\epsilon/30)^4, s)$ smooth.
In other words, there is a function $h$ such that $\mathsf{surf}(h) \le (1+(\epsilon/30)^4) \cdot s$ which is $O(\epsilon^4)$-close to $f$. 
 Further, since $f$ passes \textsf{Test-rank} with probability $0.9$, Lemma~\ref{lem:far-k-junta} implies that 
$f$ is $\epsilon^2$-close to some linear $k$-junta $g$. We now apply Lemma~\ref{lem:dual-closeness} to obtain that $f$ is $O(\epsilon)$-close to some function $\tilde{h}: \mathbb{R}^n \rightarrow \{-1,1\}$ which is a linear $k$-junta and $\mathsf{surf}(h) \le (1+O(\epsilon)) s$. This concludes the proof.

\end{proofof}}
We now turn to describing the routine \textsf{Test-rank} and prove Lemma~\ref{lem:far-k-junta}.

\begin{figure}[tb]
\hrule
\vline
\begin{minipage}[t]{0.98\linewidth}
\vspace{10 pt}
\begin{center}
\begin{minipage}[h]{0.95\linewidth}
{\small
\underline{\textsf{Inputs}}
\vspace{5 pt}

\begin{tabular}{ccl}
$s$ &:=& surface area parameter \\
$\epsilon$ &:=& error parameter \\
$k$ &:=& rank parameter\\
\end{tabular}

\vspace{5 pt}
\underline{\textsf{Testing algorithm}}
\begin{enumerate}
\item Run algorithm \textsf{Test-surface-area} with surface area parameter $s$ and error parameter $(\epsilon/30)^4$. 
\item If  \textsf{Test-surface-area} outputs \textsf{yes}, then run the algorithm \textsf{Test-rank} 
with rank parameter $k$, surface area parameter $s$ and error parameter $\epsilon$. 
\item If \textsf{Test-rank} outputs \textsf{yes}, then output \textsf{yes}. If \textsf{Test-rank} outputs \textsf{no}, output \textsf{no}. 
\end{enumerate}

\vspace{5 pt}
}
\end{minipage}
\end{center}

\end{minipage}
\hfill \vline
\hrule
\caption{Description of the   algorithm \textsf{Test-linear-junta}}
\label{fig:tlj}
\end{figure}

\begin{figure}[tb]
\hrule
\vline
\begin{minipage}[t]{0.98\linewidth}
\vspace{10 pt}
\begin{center}
\begin{minipage}[h]{0.95\linewidth}
{\small
\underline{\textsf{Input}}
\vspace{5 pt}

\begin{tabular}{ccl}
$k$ &:=& rank parameter \\
$s$ &:=& surface area parameter \\ 
$\epsilon$ &:=& error parameter 
\end{tabular}

\underline{\textsf{Parameters}}
\vspace{5 pt}

\begin{tabular}{ccl}
$t$ &:=& $\frac{\epsilon^4}{900 s^2}$ \\
$r$ &:=& $\frac{k \cdot s^2}{\epsilon^7}$\\ 
$\kappa$ &:& $\frac{\epsilon^2}{40 r}$ \\
\end{tabular}

\vspace{5 pt}
\underline{\textsf{Testing algorithm}}
\begin{enumerate}
\item Sample directions $y_1,\ldots, y_r \sim \gamma_n$. 
\item Let  $A_{i,j}=\langle D_{} P_{t} f(y_i) , D_{} P_{t} f(y_j) \rangle$. 
\item For all $1 \le i,j \le r$, compute $A_{i,j}$ up to error $\kappa$ using Lemma~\ref{lem:inner-product-1}. Call the estimates $B_{i,j}$. 
\item For the matrix $B \in \mathbb{R}^{r \times r}$,  compute the top $k+1$ singular values of $B$.
\item Output \textsf{yes} if and only if the $(k+1)^{st}$ singular value is at most $\frac{\epsilon^2}{16}$.
\end{enumerate}

\vspace{5 pt}
}
\end{minipage}
\end{center}

\end{minipage}
\hfill \vline
\hrule
\caption{Description of the \textsf{Test-rank} algorithm}
\label{fig:trj}
\end{figure}

\begin{proofof}{Lemma~\ref{lem:far-k-junta}}
The bound on the query complexity of Lemma~\ref{lem:far-k-junta} is immediate from the settings of our parameters and query complexity of Lemma~\ref{lem:inner-product-1}. 

The first item (i.e., the completeness of \textsf{Test-rank}) follows from the fact that if $f$ is a linear $k$-junta, $P_t f$ is also a linear $k$-junta. Consequently,  $A$ is a rank-$k$ matrix. Then, $A$ has at most $k$ non-zero singular values. Thus, if $\sigma_1 \ge \sigma_2 \ge \ldots$ are the singular values of $A$ (in order), then $\sigma_{k+1}=0$. By invoking Weyl's inequality (Lemma~\ref{lem:Weyl}), the $(k+1)^{th}$ singular value of $B$ is at most $\epsilon^2/10$. This finishes the proof of the first item. 

The proof of the second item (i.e., the soundness of \textsf{Test-rank}) is more involved. In particular, we can restate the second item as proving the following lemma. 
\begin{lemma}~\label{lem:far-k-junta-1}
Let $f : \mathbb{R}^n \rightarrow \{-1,1\}$  be a $((\epsilon/30)^2,s)$-smooth function which is $\epsilon$-far from a linear $k$-junta,
 then the algorithm \textsf{Test-rank} outputs \textsf{no} with probability $1-\epsilon$. 
\end{lemma}
The task of proving this lemma shall be the agenda for the rest of this section.
\end{proofof}

In order to prove Lemma~\ref{lem:far-k-junta-1}, we will need a few preliminary lemmas. The following lemma says that if a function's gradient is almost always orthogonal to a subspace $V$. Then, the function is close to a $V$-junta. 
\begin{lemma}~\label{lem:gradient-subspace}
Let $f: \mathbb{R}^n \rightarrow \mathbb{R}$ (be a $\mathcal{C}^1$ function) and let $V$ be a subspace of rank $k$ and let $W = V^{\perp}$. Let us assume that $\mathbf{E}[\Vert (Df)_W \Vert_2^2] = \epsilon$. Then there is a $V$-junta $g: \mathbb{R}^n \rightarrow \mathbb{R}$ such that $\mathbf{E}[(g(x) - f(x))^2]\le \epsilon$. 
\end{lemma}
\begin{proof}
Let us rotate the space so that $V = \{(x_1, \ldots, x_k, 0, \ldots, 0): x_1, \ldots, x_k \in \mathbb{R}\}$. 
Let us now define $g: \mathbb{R}^n \rightarrow \mathbb{R}$ as 
\[
g(x) = \mathop{\mathbf{E}}_{z \sim \gamma_{n-k}} [f(x_1, \ldots, x_k, z_1, \ldots, z_{n-k})]. 
\]
Observe that $g$ is a $V$-junta. Now, for every choice $X = (x_1, \ldots, x_k)$, consider the function $h_X: \mathbb{R}^{n-k} \rightarrow \mathbb{R}$ as 
\[
h_X(z_1, \ldots, z_{n-k}) = f(x_1, \ldots, x_k, z_1, \ldots, z_{n-k}) - g(x). 
\]
Observe that $\mathbf{E}_{(z_1, \ldots, z_{n-k}) \sim \gamma_{n-k}} [h_X(z_1, \ldots, z_{n-k})]=0$. 
By applying Lemma~\ref{lem:Poincare}, 
\[
\mathbf{E}[h_X^2(z_1, \ldots, z_{n-k})] = \mathsf{Var}[h_X(z_1, \ldots, z_{n-k})] \leq \mathbf{E}[\Vert Dh_X \Vert_2^2]. 
\]
Observe that $Dh_X(z_1, \ldots, z_{n-k}) =Df(x_1,\ldots, x_k,  z_1, \ldots, z_{n-k})_W$. Thus, we get
\begin{eqnarray*}
\mathbf{E}[(f(x) - g(x))^2] &=& \mathop{\mathbf{E}}_{X \sim \gamma_k} \mathop{\mathbf{E}}_{Z\sim \gamma_{n-k}} [h_X^2 (Z)]  \leq \mathop{\mathbf{E}}_{X \sim \gamma_k} \mathop{\mathbf{E}}_{Z\sim \gamma_{n-k}} [\Vert Dh_X \Vert_2^2] \\
&=& \mathop{\mathbf{E}}_{X \sim \gamma_k} \mathop{\mathbf{E}}_{Z\sim \gamma_{n-k}} [\Vert Df(X,Z)_W \Vert_2^2] = \epsilon.
\end{eqnarray*}
This finishes the proof. 
\end{proof}
For the rest of this section, when we use the value $t$, it will bear the same relation as stated in the description of the algorithm \textsf{Test-rank} (see Figure~\ref{fig:trj}). 
\begin{proposition}~\label{prop:noise-stab-surf-1}
Let $f: \mathbb{R}^n \rightarrow \{-1,1\}$ which is $((\epsilon/30)^2,s)$-smooth. 
 Then, 
$
\mathbf{E}[|P_{t}f - f|^2] \le \frac{\epsilon^2}{5}. 
$
\end{proposition} 
\begin{proof}
Since $f$ is $((\epsilon/30)^2,s)$ smooth, we know that there is a function $g$ such that $\mathbf{E}[|f-g|] \le (\frac{\epsilon}{30})^2$ and $\mathsf{surf}(g)\le s\big(1+(\frac{\epsilon}{30})^2\big)$. 
By using the fact that the operator $P_{t}$ is contractive, we have, 
\[
\mathbf{E}[|P_{t} f - P_{t}g|^2] \le \mathbf{E}[|f-g|^2 ]  \le 4 \mathbf{E}[\Vert f - g \Vert_1]\le \frac{\epsilon^2}{200}. 
\]
Next, we use Proposition~\ref{prop:Ledoux} to get that 
$
\mathbf{E}[|P_{t}g - g|^2 ] \le \frac{\epsilon^2}{30}. 
$
We can now combine these to get 
\[
\mathbf{E}[|P_{t} f - f|^2] = 3 \big( \mathbf{E}[|P_{t} f - P_{t} g|^2] + \mathbf{E}[|P_{t} g -  g|^2] + \mathbf{E}[| f -  g|^2]\big) \le \frac{\epsilon^2}{5}. 
\]
\end{proof}

\begin{lemma}~\label{lem:subspace-escape}
Let $f: \mathbb{R}^n \rightarrow \{-1,1\}$ be a $((\epsilon/30)^2,s)$-smooth function which is $\epsilon$-far from any linear $k$-junta. For any subspace $W$ of co-dimension at most $k$,
\[
\Pr_{y \sim \gamma_n} \bigg[\Vert D_{}P_{t}f(y) \Vert_{W}^2 \ge \frac{\epsilon^2}{8}\bigg] \ge \Omega\bigg(\frac{\epsilon^6}{s^2}\bigg). 
\]
\end{lemma} 
\begin{proof}
Applying Proposition~\ref{prop:noise-stab-surf-1}, we have that $\mathbf{E}[|P_{t}f-f|^2] \le \frac{\epsilon^2}{5}$. By applying Jensen's inequality, 
we have $\mathbf{E}[|P_{t}f-f|] \le \epsilon/\sqrt{5}$. Thus, $P_{t}f$ is $0.5\cdot \epsilon$-far from any linear $k$-junta (in $\ell_1$ distance). Consequently, we can say that for any $W$-junta $h$, $\mathbf{E}[\Vert P_{t} f - h \Vert_2^2] > 0.25 \epsilon^2$. By contrapositive of Lemma~\ref{lem:gradient-subspace}, we have that 
\begin{equation}~\label{eq:expectation}
\mathbf{E}[\Vert D_{}P_{t}f(y) \Vert_{W}^2] > 0.25 \cdot \epsilon^2. \end{equation} 

Next, observe that Lemma~\ref{lem:derivative-shift}  implies that 
\[
\Vert D_{}P_{t}f(y) \Vert_{W}^2 \le \frac{1}{e^{2t}-1} \cdot \Vert \mathcal{W}_1(f_{t,y}) \Vert_2^2 \le \frac{1}{{e^{2t}-1}}
\le O(1/t) \le O\left(\frac{s^2}{\epsilon^4}\right).
\]
The second inequality follows immediately from that $f_{t,y}$ has range bounded between $[-1,1]$. Combining this with (\ref{eq:expectation}), this implies that 
\[
\Pr\bigg[\Vert D_{}P_{t}f(y) \Vert_{W}^2 \ge \frac{\epsilon^2}{8}\bigg] \ge \Omega\bigg(\frac{\epsilon^6}{s^2}\bigg).  
\]
\end{proof}

We are now in a position to finish the proof of Lemma~\ref{lem:far-k-junta-1}.

\begin{proofof}{Lemma~\ref{lem:far-k-junta-1}}
    Let $M_i \in \mathbb{R}^n$ denote $M_i = D_{} (P_{t}f)(y_i)$. As in Figure~\ref{fig:trj}, consider the matrix $A \in \mathbb{R}^{r \times r}$ whose $(i,j)$ entry is $A_{i,j} = \langle D_{} (P_{t}f)(y_i), D_{} (P_{t}f)(y_j) \rangle$. Now, consider the matrix $M \in \mathbb{R}^{n \times k}$ whose $i^{th}$ column is $M_i$. Then, observe that $A = M^t \cdot M$. We would like to analyze the singular values of $A$. Observe that the non-zero singular values of $M^t \cdot M$ are the same as the non-zero singular values of $M \cdot M^t$. Now, observe that 
\[
M \cdot M^t = \sum_{i=1}^r D_{} (P_{t}f)(y_i) \cdot D_{} (P_{t}f)(y_i)^t 
\]
Instead of analyzing the non-zero singular values of $M^t \cdot M$, we will analyze the non-zero singular values of $M \cdot M^t$. From now on, let us use $h$ to denote  
$P_{t}f$. Let us define the sequence of stopping times $\{\tau_j \}_{j \ge 0}$ as follows: $\tau_0=0$ and let $\mathcal{G}_j = \sum_{\ell \le \tau_j}  D_{} h(y_\ell) \cdot D_{} h(y_\ell)^t$ and $W_j$ be the eigenspace formed by the top $j$ eigenvectors of $\mathcal{G}_j$. Then, $\tau_{j+1}$ is the smallest $\ell>j$ such that 
$\Vert (D_{} h(y_\ell))_{W_j^\perp} \Vert_2 \ge \frac{\epsilon}{2\sqrt{2}}$. We now make the following claim. 
\begin{claim}~\label{clm:large-singular}
For $j \le k+1$, the top $j$ singular values of $\mathcal{G}_j$ are at least $\epsilon^2/8$. 
\end{claim}
\begin{proof}
We will prove this claim by induction. So, assume that 
the top $j$ singular values of $\mathcal{G}_j$ are all at least $\epsilon^2/8$. Now, for $\ell = \tau_{j+1}$, let 
$w$ be the unit vector in the direction of the component of $D_{} h(y_\ell)$ orthogonal to $W_j$. Let $\Gamma$ be the linear span of $W_j$ and $w$. Now, consider any unit vector $v \in \Gamma$ and express it as $v= v_1+ v_2$ where $v_1$ lies in $W_j$ and $v_2$ is parallel to $w$. Next, observe that
\begin{eqnarray*}
v^T \cdot \big( \mathcal{G}_j + D_{} h(y_\ell) \cdot D_{} h(y_\ell)^t \big) \cdot v  = v^T \cdot  \mathcal{G}_j \cdot v + v^T \cdot D_{} h(y_\ell) \cdot D_{} h(y_\ell)^t \cdot v. 
\end{eqnarray*}
The first term $v^T \cdot  \mathcal{G}_j \cdot v$ is at least as large as $v_1^T \cdot  \mathcal{G}_j \cdot v_1$ and the second term $v^T \cdot D_{} h(y_\ell) \cdot D_{y} h(y_\ell)^t \cdot v$ is the same as $v_2^T \cdot D_{} h(y_\ell) \cdot D_{} h(y_\ell)^t \cdot v_2$. Next, note that
\[
v_1^T \cdot  \mathcal{G}_j \cdot v_1\ge \frac{\epsilon^2}{8} \cdot \Vert v_1\Vert_2^2 ;\ \ \ \  \ v_2^T \cdot D_{} h(y_\ell) \cdot D_{} h(y_\ell)^t \cdot v_2\ge \frac{\epsilon^2}{8} \cdot \Vert v_2\Vert_2^2.
\]
Consequently,
\[
v^T \cdot \big( \mathcal{G}_j + D_{} h(y_\ell) \cdot D_{} h(y_\ell)^t \big) \cdot v \ge \frac{\epsilon^2 }{8} \cdot \big(\Vert v_1 \Vert_2^2 + \Vert v_2 \Vert_2^2 \big) = \frac{\epsilon^2}{8}. 
\]
Observe that 
\[
v^T \cdot \mathcal{G}_{j+1} \cdot v \ge v^T \cdot \big( \mathcal{G}_j + D_{} h(y_\ell) \cdot D_{} h(y_\ell)^t \big) \cdot v \ge \frac{\epsilon^2}{8}. 
\]
The first inequality is immediate from the fact that 
\[
\mathcal{G}_{j+1} -\big( \mathcal{G}_j + D_{} h(y_\ell) \cdot D_{} h(y_\ell)^t \big) = \sum_{\tau_j <i < \tau_{j+1}} D_{} h(y_i) \cdot D_{} h(y_i)^t
\]
is a psd matrix. Thus, we obtain that 
\begin{equation}~\label{eq:singular-1}
\inf_{v: \Vert v \Vert_2=1 \ \textrm{and} \ v \in \Gamma} v^T \cdot \mathcal{G}_{j+1} \cdot v \ge \frac{\epsilon^2}{8}. 
\end{equation}
Now, it is clear that $\mathcal{G}_{j+1}$ is a psd matrix. If the singular values of $\mathcal{G}_{j+1}$ are $\sigma_1 \ge \sigma_2 \ge \ldots$, then by Courant Fischer theorem, we have 
\[
\sigma_{k+1} = \max_{S_{k+1} \subseteq \mathbb{R}^n} \inf_{v: \Vert v \Vert_2=1 \ \textrm{and} \ v \in S_{k+1}} v^T \cdot \mathcal{G}_{j+1} \cdot v,
\]
where $S_{k+1}$ is the set of all $k+1$ dimensional subspaces of $\mathbb{R}^n$. Thus, by applying (\ref{eq:singular-1}) and observing $\mathsf{dim}(\Gamma) = k+1$, we get 
\[
\sigma_{k+1} \ge \inf_{v: \Vert v \Vert_2=1 \ \textrm{and} \ v \in \Gamma} v^T \cdot \mathcal{G}_{j+1} \cdot v \ge \frac{\epsilon^2}{8}.
\] 
This finishes the proof. 
\end{proof}
Now applying Lemma~\ref{lem:subspace-escape}, we have that  conditioned on $\tau_j$, $\tau_{j+1}-\tau_j$ is a geometric random variable with parameter (at least) $\Omega(\epsilon^6/s^2)$. From this, it is not difficult to see that with probability at least $1-\epsilon$, $\tau_{k+1} = O(s^2 \cdot k/\epsilon^7)$, 
Thus, with probability $1-\epsilon$, we can assume that the top $k+1$ singular values of $M \cdot M^t$ are all at least $\epsilon^2/8$.

Consequently, we get that the top $k+1$ singular values of $A = M^t \cdot M$ are all at least $\epsilon^2/8$. Now, the algorithm computes a matrix $B$ such that 
$\Vert A-B\Vert_F \le \epsilon^2/10$.  By Weyl's inequality~(Lemma~\ref{lem:Weyl}), we get that the top $k+1$ singular values of $B$ are all at least $\epsilon^2/16$. This proves the lemma. 

\end{proofof}
{
We finally give the proof of Lemma~\ref{lem:dual-closeness}. The proof relies on the so-called co-area formula. 
\begin{lemma}~\label{lem:coarea}
   Let $f: \R^n \to [-1, 1]$ be smooth and $\psi: [-1, 1] \to \R_+$ be bounded and measurable. Then
    \[
        \int_{-1}^1 \psi(s) \mathsf{surf}(\{x: f(x) \le s\}) \, ds = \int_{\R^n} \psi(f(x)) |\nabla f(x)| \, d\gamma(x).
    \]
\end{lemma}
\begin{proofof}{Lemma~\ref{lem:dual-closeness}}
    By~\cite{maggi2012sets}, there is a smooth function $h_1: \R^n \to [-1, 1]$ with bounded gradient
    such that $\|h_1 - h\|_2 \le \epsilon$ and $\E[|\nabla h_1|] \le 2s$.
    Let $E$ be a $k$-dimensional subspace for which $g$ is an $E$-junta, and let $z$ be a standard Gaussian
    vector on $E^\perp$. Let $\Pi_E$ be the projection operator for subspace $E$ and define $h_2: \R^n \to [-1, 1]$ by $h_2(x) = \E_z[h_1(\Pi_E x + z)]$. By Jensen's inequality,
    $\E [|\nabla h_2|] \le \E[|\nabla h_1|] \le 2s$. Let $t$ be uniformly distributed in $[-1+\eta, 1-\eta]$,
    and define $\tilde h = \tilde h_t$ by
    \[
        \tilde h_t(x) = \tilde h(x) = \begin{cases}
            -1 &\text{if $h_2(x) \le t$} \\
            1 &\text{otherwise}.
        \end{cases}
    \]
    Note that $\tilde h_t$ is an $E$-junta (because $h_2$ is an $E$-junta).
    In expectation over $t$, the surface area of $\tilde h$ is
    \[
        \frac{1}{2-2\eta} \int_{-1+\eta}^{1-\eta} \mathsf{surf}\{x: h_2 \le s\}\, ds,
    \]
    which by the co-area formula is equal to
    \[
        \frac{1}{2-2\eta} \int_{\R^n} 1_{\{h_2(x) \in [-1+\eta, 1-\eta]\}} |\nabla h_2(x)|\, d\gamma(x)
        \le \frac{1}{2-2\eta} \E_{x \sim \gamma} [|\nabla h_2(x)|]
        \le \frac{s}{1-\eta}.
    \]
    In particular, there exists some $t \in [-1+\eta, 1-\eta]$ such that the surface area
    of $\tilde h_t$ is at most $\frac{s}{1-\eta}$.

    Next, we will estimate the distance of $\tilde h$ from $h$. By the triangle inequality, $\|h - g\|_2 \le 2\epsilon$
    and so $\|h_1 - g\|_2 \le 3\epsilon$. On the other hand, Pythagoras' theorem implies
    that $h_2$ minimizes $\|h_1 - h_2\|_2$ among all $E$-juntas; hence, $\|h_1 - h_2\|_2 \le 3\epsilon$
    and so $\|h - h_2\| \le 4 \epsilon$.
    Now, $h$ takes values in $\{-1, 1\}$ and so $|h(x) - h_2(x)| \ge \eta 1_{\{h_2(x) \in [-1+\eta, 1-\eta]\}}$.
    On the other hand, the definition of $\tilde h$ ensures that
    \[
        |\tilde h(x) - h_2(x)| \le \begin{cases}
            2 &\text{if $h_2(x) \in [-1+\eta, 1-\eta]$} \\
            \eta &\text{otherwise}.
        \end{cases}
    \]
    If $p$ is the probability that $h_2(x) \in [1 + \eta, 1-\eta]$, it follows that $\eta \sqrt{p} \le \|h - h_2\|_2$
    and so
    \[
        \|\tilde h - h_2\|_2 \le 2\sqrt{p} + \eta \le \frac{8\epsilon}{\eta} + \eta.
    \]
    By the triangle inequality $\|\tilde h - h\|_2 \le 4 \epsilon + \frac{8\epsilon}{\eta} + \eta$.
    Choosing $\eta = \sqrt \epsilon$ completes the proof.
\end{proofof}
}

\section{Algorithm to find hidden linear invariant structure}~\label{aff:inv}
In this section, we will prove the following main theorem. 
\begin{theorem}~\label{thm:main-affine-invariant}
Let $f: \mathbb{R}^n \rightarrow \{-1,1\}$ be a linear-$k$-junta with surface area $s$. Then, there is an algorithm \textsf{Find-invariant-structure} which for any error parameter $\epsilon>0$, 
makes $O(s \cdot k /\epsilon)^{O(k)}$ queries to $f$
and with probability $1-\epsilon$  outputs (for some $\ell \le k$) a function $g: \mathbb{R}^\ell \rightarrow [-1,1]$ so that the following holds:  there is an orthonormal set of vectors $w_1, \ldots, w_\ell \in \mathbb{R}^n$ such that 
$$
\mathbf{E}[|f(x) - g(\langle w_1, x\rangle, \ldots, \langle w_\ell, x \rangle)|] = O(\epsilon). 
$$
 {Further, there is a set $V = \{v_1, \ldots, v_k\}$ of orthonormal vectors such that for $1 \le j \le \ell$, $v_j = w_j$ and $\span\{v_1, \ldots, v_k\}$ is a relevant subspace of $f$.}
\end{theorem}

Our algorithm is quite na\"ive. First, we ``identify'' -- in some implicit
sense -- the $k$-dimensional subspace on which the linear $k$-junta acts. We
take a fine net of functions defined on that space, and we test them all until
we fine the best one. Obviously, this algorithm is not computationally
efficient, and it is also not particularly efficient in terms of the query
complexity. However, the crucial feature of this algorithm
is that its query complexity does not depend on the
ambient dimension $n$. The main difficulty in constructing and analyzing this algorithm
is that we cannot explicitly identify even a single vector in the interesting
$k$-dimensional subspace -- that would require a number of queries that depends on $n$.
One consequence of this is that we do not know how to apply an off-the-shelf
learning algorithm (such as the one from~\cite{KOS:08}).

\begin{definition}~\label{def:vector-independence}
A set of vectors $v_1, \ldots, v_\ell \in \mathbb{R}^n$ is said to be $(\eta,\gamma)$-linearly independent if the following conditions hold: 
\begin{enumerate}
\item For all $1 \le i \le \ell$, 
$\Vert v_i \Vert_2 \le \eta$. 
\item For all $1 < i \le \ell$, $\mathsf{dist}(v_i, \mathsf{span}(v_1, \ldots, v_{i-1})) \ge \gamma$. 
\end{enumerate}
\end{definition}

\begin{definition}
For $f: \mathbb{R}^n \rightarrow [-1,1]$ and $t>0$, we say that a set of directions $(y_1, \ldots, y_\ell)$ is 
$\gamma$-linearly independent, if the following holds: 
For $1 \le i \le \ell$, let $v_i = DP_tf(y_i)$. If for all $i$, $\mathsf{dist}(v_i, \mathsf{span}(v_1, \ldots, v_{i-1})) \ge \gamma$. 
\end{definition}
By Proposition~\ref{prop:derivative-bound}, it is immediate that as long as $t \le 1/4$, $\Vert DP_t f (y) \Vert_2 \le t^{-1/2}$. Thus, if $(y_1, \ldots, y_\ell)$ is $\gamma$-linearly independent, then the directions $(v_1,\ldots, v_\ell)$ are $(t^{-1/2}, \gamma)$ linearly independent. 

\begin{figure}[h]
\hrule
\vline
\begin{minipage}[t]{0.98\linewidth}
\vspace{10 pt}
\begin{center}
\begin{minipage}[h]{0.95\linewidth}
{\small
\underline{\textsf{Inputs}}
\vspace{5 pt}

\begin{tabular}{ccl}
$t$ &:=& noise parameter \\
$y_1, \ldots, y_\ell$ &:=& $\frac{\gamma}{2}$-linearly independent directions \\
$\{\beta_{i,j}\}$ &:=& $\lambda$-accurate estimates of $\langle DP_tf(y_i), DP_tf(y_j)\rangle$ where \\
&& $\lambda=\lambda(\ell,\nu, t^{-\frac12}, \gamma/2)$ and $\nu = \frac{\gamma^2 \cdot t}{100 \ell^2}$ (from Lemma~\ref{prop:linear}) 
\\
$y_{\ell+1}$ &:=& candidate direction in $\mathbb{R}^n$. 
\end{tabular}

\vspace{5 pt}
\underline{\textsf{Testing algorithm}}
\begin{enumerate}
\item Find the numbers $\{\alpha_{1 \le i, j \le \ell}\}$ from Lemma~\ref{prop:linear}. 
\item Estimate $\langle DP_tf(y_{\ell+1}) ,  DP_tf(y_{\ell+1}) \rangle$ up to $\pm \frac{\gamma^2}{50}$ . Call the estimate $\tilde{\beta}_{\ell+1, \ell+1}$. 
\item Estimate $\langle DP_tf(y_{\ell+1}) ,  DP_tf(y_{j}) \rangle$ (for $1 \le j \le \ell$) up to accuracy 
$\frac{1}{\xi(\ell, t^{-1/2}, \gamma/2)} \cdot \frac{\gamma^2 \cdot \sqrt{t}}{100\ell^3}$ (using  Lemma~\ref{lem:inner-product-1}) where $\xi$ is the function from  Lemma~\ref{prop:linear}.  Call the estimates $\tilde{\beta}_{j, \ell+1}$. 
\item Compute quantity $\zeta_i = \sum_{1 \le j \le \ell} \alpha_{i,j} \cdot \tilde{\beta}_{j,\ell+1}$ for all $1 \le i \le \ell$. 
\item If the quantity $\tilde{\beta}_{\ell+1, \ell+1}^2 - \sum_{i=1}^\ell \zeta_i^2 > (\frac{3 \gamma}{4})^2$, then output \textsf{yes}. Else output \textsf{no}. 
\end{enumerate}

\vspace{5 pt}
}
\end{minipage}
\end{center}

\end{minipage}
\hfill \vline
\hrule
\caption{Description of the   algorithm \textsf{Test-candidate-direction}}
\label{fig:tlin-1}
\end{figure}

\begin{lemma}~\label{lem:test-candidate}
The algorithm \textsf{Test-candidate-direction} described in Figure~\ref{fig:tlin-1} has the following properties: For noise parameter $t$, directions $y_1, \ldots, y_\ell \in \mathbb{R}^n$,  $\{\beta_{i,j} \}$ and candidate direction $y_{\ell+1}$ (where $y_1, \ldots, y_\ell$ as well as $\{ \beta_{i,j} \}$ meet the requirements described in Figure~\ref{fig:tlin-1}), the algorithm satisfies
\begin{enumerate}
\item The query complexity of the algorithm is 
$T_{tc}(t, \gamma, \ell) = \big( \frac{\ell}{\sqrt{t} \cdot \gamma} \big)^{O(\ell)}$. 
\item If the Euclidean distance of $DP_tf(y_{\ell+1})$ is at least $\gamma$ from the subspace $\mathsf{span}(DP_tf(y_{1}), \ldots, DP_tf(y_{\ell}))$, then the algorithm outputs \textsf{yes}. Conversely, if the algorithm outputs \textsf{no}, then the Euclidean distance must be less than $\frac{\gamma}{2}$.  
\end{enumerate}
\end{lemma}

\begin{proof}
The query complexity bound is just immediate from Lemma~\ref{lem:inner-product-1} and plugging in the value of $\xi(\ell, t^{-1/2}, \gamma)$ from Lemma~\ref{prop:linear}. To prove the second guarantee, 
let us use $v_j$ to denote $DP_tf(y_j)$. Since $(v_1, \ldots, v_j)$ are $(1/t^{-1/2}, \frac{\gamma}{2})$-linearly independent, hence by Lemma~\ref{prop:linear}, we obtain that there are orthonormal vectors $(w_1, \ldots, w_\ell)$ (which span $v_1, \ldots, v_\ell$) 
such that 
$$
\Vert w_i - \sum_{j} \alpha_{i,j} v_j \Vert_2 \le \frac{\gamma^2 \cdot {t}}{100 \ell^2}. 
$$
This implies that if we let $v_{\ell+1} = DP_tf(y_{\ell+1})$ (using $\Vert v_{\ell+1} \Vert \le t^{-1/2}$), then 
$$
\big| \langle w_i, v_{\ell+1} \rangle - \sum_{j} \alpha_{i,j} 
\langle v_j , v_{\ell+1} \rangle \big| \le \frac{\gamma^2 \sqrt{t}}{100 \ell^2}. 
$$
Consequently, we have 
\begin{eqnarray*}
\big| \langle w_i, v_{\ell+1} \rangle - \sum_{j} \alpha_{i,j} \cdot \tilde{\beta}_{j,\ell+1} \big| &\le& \frac{\gamma^2 \cdot \sqrt{t}}{100 \ell^2} + \sum_{j}|\alpha_{i,j}| \cdot |\langle v_j , v_{\ell+1} \rangle - \tilde{\beta}_{j,\ell+1} |  \\
&\le& \frac{\gamma^2\cdot \sqrt{t}}{100 \ell^2} + \sum_{j} \xi (\ell, t^{-1/2}, \gamma/2) \cdot \frac{1}{\xi (\ell, t^{-1/2}, \gamma/2)}\cdot \frac{\gamma^2\sqrt{t}}{100 \ell^3} \le \frac{\gamma^2\sqrt{t}}{50 \ell^2}. 
\end{eqnarray*}
The penultimate inequality follows from the bound on $|\alpha_{i,j}|$ from Lemma~\ref{prop:linear} and the accuracy of estimates $\tilde{\beta}_{j,\ell+1}$. 
This implies that for any $i$, 
\begin{equation}~\label{eq:bound-diff1}
\big|\big| \langle w_i, v_{\ell+1} \rangle \big|^2- \big|\sum_{j} \alpha_{i,j} \cdot \tilde{\beta}_{j,\ell+1} \big|^2\big|  \le \frac{\gamma^2\sqrt{t}}{50 \ell^2} \cdot \big| \langle w_i, v_{\ell+1} \rangle + \sum_{j} \alpha_{i,j} \cdot \tilde{\beta}_{j,\ell+1} \big| \le \frac{\gamma^2\sqrt{t}}{50 \ell^2} \cdot 2 \cdot t^{-\frac12} = \frac{\gamma^2}{25 \ell^2}.
\end{equation}
The second inequality uses that fact that $w_i$ is a unit vector whereas $\Vert v_{\ell+1} \Vert_2 \le t^{-\frac12}$.  Thus, 
\begin{eqnarray*}
\dist^2\big(DP_tf(y_{\ell+1}), \mathsf{span}(DP_tf(y_{1}), \ldots, DP_tf(y_{\ell}))\big) &=& \Vert DP_tf(y_{\ell+1})\Vert_2^2 - \sum_{j=1}^\ell \langle DP_tf(y_{\ell+1}), w_j\rangle^2 \\
&=&  \Vert DP_tf(y_{\ell+1})\Vert_2^2 - \sum_{j=1}^\ell \zeta_j^2 + \theta
\end{eqnarray*}
where $|\theta| \le \frac{\gamma^2}{25\ell}$ (from \ref{eq:bound-diff1}). Using the fact that $ |\tilde{\beta}_{\ell+1, \ell+1}^2-  \Vert D P_tf(y_{\ell+1})\Vert_2^2 | \le \frac{\gamma^2}{50}$, we can conclude that 
$$
\big|\dist^2\big(DP_tf(y_{\ell+1}), \mathsf{span}(DP_tf(y_{1}), \ldots, DP_tf(y_{\ell}))\big)-  \tilde{\beta}_{\ell+1, \ell+1}^2 - \sum_{i=1}^\ell \zeta_i^2\big| \le \frac{\gamma^2}{25}. 
$$
Item 2 in the claim is now an immediate consequence. 
\end{proof}

\begin{figure}[tb]
\hrule
\vline
\begin{minipage}[t]{0.98\linewidth}
\vspace{10 pt}
\begin{center}
\begin{minipage}[h]{0.95\linewidth}
{\small
\underline{\textsf{Inputs}}
\vspace{5 pt}

\begin{tabular}{ccl}
$s$ &:=& surface parameter \\
$\epsilon$ &:=& error parameter\\
\end{tabular}
~\\
~\\
\underline{\textsf{Parameters}}
\vspace{5 pt}

\begin{tabular}{ccl}
$t$ &:=& $\frac{\epsilon^4}{900 s^2}$ \\
$\gamma$ &:=& $\frac{\epsilon^2}{8}$\\
$\lambda$ &=& $\lambda(k, \nu, t^{-\frac12}, \gamma)$ (where $\lambda(\cdot)$ is the function from Lemma~\ref{prop:linear}) and $\nu = \frac{\gamma^2 \cdot t}{100k^2}$. \\
$\tau_{\mathsf{succ}}$  &:=& $\frac{\epsilon^6}{s^2}$ \\
$T_{\mathsf{succ}}$  &:=& $\frac{1}{\tau_{\mathsf{succ}}} \cdot \log (10k/\epsilon)$.\\
\end{tabular}
~\\
~\\
\underline{\textsf{Testing algorithm}}
\begin{enumerate}
\item Initialize $S$ to be the empty set. 
\item Initialize $\mathsf{count}=0$. 
\item If $\mathsf{count} =k$, exit; 
\item else set $S =\{y_1, \ldots, y_\ell\}$ and compute $\{\beta_{i,j}\}$ as $\lambda$-accurate estimates of $\langle DP_tf(y_i), DP_tf(y_j) \rangle$ (Lemma~\ref{lem:inner-product-1}).  
\item Repeat $T_{\mathsf{succ}}$ times  
\item \hspace{7pt} Choose $z \sim \gamma_n$. 
\item \hspace{7pt} Run \textsf{Test-candidate-direction} with $S=\{y_1, \ldots,y_\ell\}$, candidate direction $z$,  $\gamma, t$ as defined in \textsf{Parameters} and $\{\beta_{i,j}\}$ \hspace{3pt} as computed above. 
\item \hspace{7pt} If \textsf{Test-candidate-direction} outputs \textsf{yes}, 
 add $z$ to $S$;  $\mathsf{count}+=1$; 
 go to step 3; 
 \item If the size of $S$ does not increase in $T_{\mathsf{succ}}$ steps, then exit; 
\end{enumerate}

\vspace{5 pt}
}
\end{minipage}
\end{center}

\end{minipage}
\hfill \vline
\hrule
\caption{Description of the   algorithm \textsf{Find-candidate-directions}}
\label{fig:flin-1}
\end{figure}
We now give an algorithm which  finds out directions $\{y_1, \ldots, y_\ell\}$ such that for $t$ defined before (as $t : = \frac{\epsilon^4}{900 s^2}$), $P_t f$ is close to a junta on the directions $\{DP_tf(y_1), \ldots, DP_tf(y_\ell)\}$. 
\begin{lemma}~\label{lem:find-dirs}
The algorithm \textsf{Find-candidate-directions} described in Figure~\ref{fig:flin-1} has the following properties: For noise parameter $t$, error parameter $\epsilon$, surface area parameter $s$, if the function $f: \mathbb{R}^n \rightarrow [-1,1]$ has surface area $s$ and is a linear $k$-junta, then with probability $1-\epsilon$, the algorithm outputs vectors $y_1, \ldots, y_\ell \in \mathbb{R}^n$ ($\ell \le k$) such that for $\{v_1, \ldots, v_\ell\}$ defined as 
$v_i = DP_tf(y_i)$, the function is $\epsilon$-close to 
a junta on $\mathsf{span}(v_1, \ldots, v_\ell)$. Further, the directions $(y_1, \ldots, y_\ell)$ are at
least $\gamma/2 = \frac{\epsilon^2}{16}$ linearly independent. 
 The query complexity of this algorithm is $T_{fc} (s,k,\epsilon) =\big( \frac{s \cdot k}{\epsilon}\big)^{O(k)}$. 
\end{lemma}
\begin{proof}
{{The bound on the query complexity of this algorithm is immediate by just plugging in the query complexity of the routine \textsf{test-candidate-direction} (Lemma~\ref{lem:test-candidate}) and the query complexity of Step  4~(Lemma~\ref{lem:inner-product-1}).} }

Next, observe that by the guarantee of \textsf{Test-candidate-direction}, the set $S$ output by the algorithm consists of $\gamma/2$-linearly independent 
directions. 
Finally, assume that $f$ is a $W$-junta where $\mathsf{dim}(W) \le k$. Then, note that for any $y \in \mathbb{R}^n$, $DP_tf(y) \in W$. Now, there are two  possibilities:  (For the rest of this proof, we will use $v_i$ as a shorthand for $DP_tf(y_i)$)
\begin{itemize}
\item[(a)] If $\mathsf{count}=k$, then note that we have found $k$ directions $y_1, \ldots, y_k$ such that $v_i\in W$. Further, the directions $(v_1, \ldots, v_k)$ are $(t^{-1/2}, \gamma)$-linearly independent. Thus, $\mathsf{span}(v_1, \ldots, v_k)= W$. So, in this case, $P_tf$ is indeed a junta on $\mathsf{span}(v_1, \ldots, v_k)$ (where $S= \{y_1, \ldots, y_k\}$). 
\item[(b)] If $\mathsf{count}<k$, then we are in one of the two situations: either $f$ is $\epsilon$-close to a junta on $\mathsf{span}(v_1, \ldots, v_\ell)$ where $S=\{y_1, \ldots, y_\ell\}$. In this case, we are already done. If not, then we apply Lemma~\ref{lem:subspace-escape} and obtain that with probability at least $\tau_{\mathsf{succ}}$, a randomly chosen direction $z$  will be at least $\gamma=\epsilon^2/8$-far from the subspace $\mathsf{span}(v_1, \ldots, v_\ell)$ and will thus pass the algorithm \textsf{Test-candidate-direction}. Thus, over $T_{\mathsf{succ}}$ trials, with probability at least $1- \frac{\epsilon}{10k}$, the set $S$ will increase in size and we will continue inductively. 
 Since the outer loop (i.e., the loop for $\mathsf{count}$ will run at most $k$ times), the total probability that $P_t f$ is not $\epsilon$-close to a $W$-junta for $W = \mathsf{span}(v_1, \ldots, v_\ell)$ but the algorithm terminates is at most $1-\frac{\epsilon}{10}$. This finishes the proof. 
\end{itemize}
\end{proof}
With the aid of the algorithm \textsf{Find-candidate-directions}, we are able to find implicitly find directions $\{v_1, \ldots, v_\ell\}$ such that $P_t f$ is close to a junta on $\mathsf{span}(v_1, \ldots, v_\ell)$. In the next subsection, we essentially do a hypothesis testing over a set of functions which form a cover for all juntas on $\mathsf{span}(v_1, \ldots, v_\ell)$. 

\subsection{Hypothesis testing against subspace juntas} 
The following lemma says how given the directions $y_1, \ldots, y_\ell$ and an error parameter $\tau$, we can implicitly find directions which form an orthonormal 
basis of $\mathsf{span}(v_1,\ldots, v_\ell)$ (as before, we are using $v_1, \ldots, v_\ell$ as a shorthand for $DP_tf(y_1), \ldots, DP_tf(y_\ell)$ respectively). All the symbols below will have the same value as Lemma~\ref{lem:find-dirs} unless mentioned otherwise. 
\begin{lemma}~\label{lem:orthogonalize}
Choose any error parameter $\tau>0$ and let $y_1, \ldots, y_\ell$ be $\gamma/2$-linearly independent directions for $P_t f$. Then, there is a procedure \textsf{Compute-ortho-transform} which makes  $T_{\mathsf{ortho}} =\mathsf{poly}(1/\tau) \cdot \big( \frac{\ell}{\gamma \cdot t}\big)^{O(\ell)}$ queries to $f$, we can obtain numbers $\{\alpha_{i,j} \}_{1 \le i,j \le \ell}$ such that the following holds:  
\begin{enumerate}
\item For $\Lambda(\ell, t, \gamma) = (\frac{\ell}{t \gamma})^{O(\ell)}$, all the numbers $|\alpha_{i,j}| \le \Lambda(\ell, t, \gamma)$. 
\item There exists an orthonormal basis $(w_1, \ldots, w_\ell)$ of $\mathsf{span}(v_1, \ldots, v_\ell)$ such that for all $1 \le i \leq \ell$, 
\[
\Vert w_i  - \sum_{j} \alpha_{i,j} v_j \Vert_2 \le \tau. 
\]
\end{enumerate}
\end{lemma}
\begin{proof}
Let $\lambda(\cdot)$ be the function defined in Lemma~\ref{prop:linear}. Now, observe that 
\[
\lambda(\ell, \tau, t^{-1/2}, \gamma) = \tau \cdot \bigg(\frac{\gamma \cdot t}{2 \cdot \ell } \bigg)^{O(\ell)}. 
\]
Thus, using Lemma~\ref{lem:inner-product-1}, we can use 
$T_{\mathsf{ortho}}$ queries to $f$ to obtain numbers $\{\beta_{i,j}\}_{1 \le i,j \le \ell}$ such that 
\[
\big| \beta_{i,j} - \langle D_{y_i} h(y_i) , D_{y_j} h(y_j) \rangle \big| \le \lambda(\ell, \tau, t^{-1/2}, \gamma). 
\]
As $(y_1, \ldots, y_\ell)$ are $\gamma$-linearly independent, hence the vectors $(v_1, \ldots, v_\ell)$ are $(t^{-\frac12}, \gamma)$-linearly independent. 
With this, we can now apply Lemma~\ref{prop:linear} to obtain numbers $\{\alpha_{i,j}\}$ such that there is an orthonormal basis $(w_1, \ldots, w_\ell)$ of $\mathsf{span}(DP_tf(y_1), \ldots, DP_tf(y_\ell))$  with the property that (a) 
 $
\Vert w_i  - \sum_{j} \alpha_{i,j} v_j \Vert_2 \le \tau$ 
and (b)
$|\alpha_{i,j}| \le \Lambda(\ell, t,\gamma)$ where $\Lambda(\ell, t, \gamma) = (\frac{\ell}{t \gamma})^{O(\ell)}$. 
\end{proof}

\begin{figure}[tb]
\hrule
\vline
\begin{minipage}[t]{0.98\linewidth}
\vspace{10 pt}
\begin{center}
\begin{minipage}[h]{0.95\linewidth}
{\small
\underline{\textsf{Inputs}}
\vspace{5 pt}

\begin{tabular}{ccl}
$s$ &:=& surface parameter \\
$\epsilon$ &:=& error parameter\\
$y_1, \ldots, y_\ell$ &:=& $\frac{\gamma}{2}$-linearly independent directions for $P_t f$\\
%
\end{tabular}
~\\
~\\
\underline{\textsf{Parameters}}
\vspace{5 pt}

\begin{tabular}{ccl}
$t$ &:=& $\frac{\epsilon^4}{900 s^2}$ \\
$\gamma$ &:=& $\frac{\epsilon^2}{8}$\\
$\tau$ &=& $\frac{\epsilon^2 \cdot \sqrt{t}}{100  \cdot \ell^{3/2}}$  \\
$\delta$  &:=& $\frac{\epsilon}{10}$ \\
$K$ &:=& $\ell^2 \cdot \Lambda(\ell,t,\gamma)$ where $\Lambda(\cdot)$ is defined in Lemma~\ref{lem:orthogonalize}. \\
$\xi$ &:=& $\frac{\epsilon^2 \cdot \sqrt{t}}{K \cdot \ell^3 }$\\
$\mu$ &:=& $\frac{\epsilon}{|\mathsf{Cover}(t,\ell,\delta)|}$ where $\mathsf{Cover}(\cdot, \cdot, \cdot)$ is the set from Theorem~\ref{thm:net}. \\
$J$ &:=& $\frac{10}{\epsilon^2} \cdot \log(1/\mu)$ \\
\end{tabular}
~\\
~\\
\underline{\textsf{Testing algorithm}}
\begin{enumerate}
\item Run the procedure \textsf{Compute-ortho-transform} with directions $(y_1,\ldots, y_\ell)$ and $\gamma$, $t$ and $\tau$ as set above. 
\item Let the output be parameters $\{\alpha_{i,j}\}_{1 \le i, j \le \ell}$. 
\item Sample $J$ points from $\gamma_n$. Call the points $x_1$, $\ldots$, $x_J$. 
\item For each of the points $x_i$ and each direction $y_j$, 
\item \hspace{7pt} Compute the function $f_{\partial, \xi, t, y_j}(x_i)$ (up to error $\xi$) using Lemma~\ref{lem:compute-derivative-x}. Call this $\zeta_{i,j}$. 
\item \hspace{7pt} Compute $\overline{x}_{i,j'} = \sum_{j}\alpha_{j',j} \cdot \zeta_{i,j}$.  
\item For all $g \in \mathsf{Cover}(t, \ell, \delta)$,  
compute $\mathcal{O}_g= \frac1s \cdot \sum_{i=1}^s |P_tf(x_i) - g(\overline{x}_{i,1}, \ldots, \overline{x}_{i,\ell})|$.
\item Return the $g$ which has the smallest value of $\mathcal{O}_g$.  
 
\end{enumerate}

\vspace{5 pt}
}
\end{minipage}
\end{center}

\end{minipage}
\hfill \vline
\hrule
\caption{Description of the   algorithm \textsf{Estimate-closest-hypothesis}}
\label{fig:hyp}
\end{figure}
 Let us now again set the parameters $t$ and $\gamma$ exactly the same as Lemma~\ref{lem:find-dirs}. Namely, we set $t= \frac{\epsilon^4}{900s^2}$ and $\gamma=\frac{\epsilon^2}{8}$. With this setting of parameters, we state the following lemma. 
\begin{lemma}~\label{lem:test-hypothesis}
There is an algorithm \textsf{Estimate-closest-hypothesis} (described in Figure~\ref{fig:hyp})  which takes as input  oracle access to $f: \mathbb{R}^n \rightarrow \{-1,1\}$,  directions $(y_1, \ldots, y_\ell)$ which are $\gamma/2$-linearly independent,  error parameter $\epsilon$, 
surface area parameter $s$.  The algorithm has the following guarantee: 
\begin{enumerate}
\item It makes $O\big( \frac{s \cdot \ell}{\epsilon}\big)^{O(\ell)}$ queries to $f$. 
\item There is an orthonormal basis $(w_1, \ldots, w_\ell)$ of $\mathsf{span}(DP_tf(y_1), \ldots, DP_tf(y_\ell))$ (which is independent of $g$) such that with probability $1-\epsilon$, outputs a function $g: \mathbb{R}^\ell \rightarrow [-1,1]$ with the following guarantee: Let $\mathsf{Cover}(t,\ell, \delta)$ be the set of functions from Theorem~\ref{thm:net} where the parameters $t, \delta$ are set as in Figure~\ref{fig:hyp}.  Then, 
\[
\mathbf{E}[|P_tf(x) - g(\langle w_1,x\rangle, \ldots, \langle w_\ell,x\rangle)|] \le \min_{g^\ast \in \mathsf{Cover}(t,\ell,\delta)} \mathbf{E}[|P_tf(x) - g^\ast(\langle w_1,x\rangle, \ldots, \langle w_\ell,x\rangle)|] + 5\epsilon.
\]
\end{enumerate}
\end{lemma}
\begin{proof}
As usual, the query complexity of the procedure is easily seen to be $O\big( \frac{s \cdot \ell}{\epsilon}\big)^{O(\ell)} $ by just plugging in the values of the parameters along with the guarantees on the query complexity of \textsf{Compute-ortho-transform} (Lemma~\ref{lem:orthogonalize}) as well Lemma~\ref{lem:compute-derivative-x}. 

To analyze the algorithm, let us now define a point $x \in \mathbb{R}^n$ to be \emph{good} if the following two conditions hold: 
\begin{enumerate}
\item For $1 \le i \le \ell$, 
\[
\big|f_{\partial, \xi, t, y_i}(x) - \langle DP_tf (y_i), x\rangle \big| \le \frac{\ell \cdot \xi}{\epsilon}. 
\]
\item For all $1 \le i \le \ell$, 
\[
\big| \sum_j \alpha_{i,j} \langle DP_tf(y_j), x\rangle - \langle w_i, x\rangle \big| \le \frac{\epsilon \cdot \sqrt{t}}{100 \ell^2}. 
\]
\end{enumerate} 
\begin{claim}
For $x \sim \gamma_n$, $\Pr[x \textrm{ is good}] \ge 1-\frac{2\epsilon^2}{\ell}$. 
\end{claim}
\begin{proof}
Lemma~\ref{lem:compute-derivative-x} guarantees that for any specific choice of $i$, 
$\Pr[\big|f_{\partial, \xi, t, y_i}(x) - \langle DP_tf (y_i), x\rangle \big| \le \frac{\ell \cdot \xi}{\epsilon}] \le \frac{\epsilon^2}{\ell^2}$. Thus, with probability $1-\frac{\epsilon^2}{\ell}$, item $1$ holds for $x \sim \gamma_n$. Likewise, notice that 
$$
\Vert \sum_j \alpha_{i,j} DP_tf(y_j)- w_i \Vert_2 \le \tau. 
$$ 
Thus, for any $x_i \sim \gamma_n$, with probability $1- \frac{\epsilon^2}{\ell^2}$, item 2 holds. Thus, by a union bound, it holds for all $1 \le i\le \ell$ simultaneously, with probability $1-\frac{\epsilon^2}{\ell}$. This proves the claim. 
\end{proof}
Next, observe that if a point $x_i$ is \emph{good}, then the following holds for every $j'$: 
\begin{eqnarray}
\big|\overline{x}_{i,j'} - \langle w_{j'}, x_i \rangle\big| &\leq& \sum_{j} \big| \alpha_{j',j} \cdot \zeta_{i,j} - \alpha_{j',j}  \cdot f_{\partial, \xi, t, y_j}(x_i) \big| + \big| \langle w_j', x_i \rangle - \sum_{j} \alpha_{j',j}  \cdot f_{\partial, \xi, t, y_j}(x_i) \big| \nonumber \\
&\le& \xi \cdot \sum_{j} |\alpha_{j',j}|  + \big| \langle w_j', x_i \rangle - \sum_{j} \alpha_{j',j}  \cdot f_{\partial, \xi, t, y_j}(x_i) \big| \nonumber\\ &\le&  \frac{\epsilon^2\sqrt{t}}{2 \ell^4 }+ \big| \langle w_j', x_i \rangle - \sum_{j} \alpha_{j',j}  \cdot f_{\partial, \xi, t, y_j}(x_i) \big|  \nonumber \\ 
&\le&  \frac{\epsilon^2\sqrt{t}}{2 \ell^4 } + 
\big| \langle w_j', x_i \rangle - \sum_{j} \alpha_{j',j} \langle DP_tf(y_j), x_i \rangle \big|   + 
\sum_{j} |\alpha_{j',j}| \cdot \big| \langle DP_tf(y_j), x_i - f_{\partial, \xi, t, y_j}(x_i) \big|
\nonumber \\
&\le& \frac{\epsilon^2\sqrt{t}}{2 \ell^4 } + \frac{\epsilon^2 \sqrt{t}}{100 \cdot \ell^2}  + \frac{\epsilon \sqrt{t}}{100 \cdot \ell^2} \le \frac{\epsilon \cdot \sqrt{t}}{\ell^2}. ~\label{eq:good}
\end{eqnarray}
The penultimate inequalities just follow from the condition that $x_i$ is \emph{good} and the values of the parameters. Now, observe that 
\begin{eqnarray}
&& \big| \mathbf{E}_{x\sim \gamma_n} [|P_tf(x) - g(\overline{x}_1,\ldots, \overline{x}_\ell)|] - \mathbf{E}_{x\sim \gamma_n} [|P_tf(x) - g(\langle w_1,x \rangle, \ldots, \langle w_\ell, x \rangle)|] \big| \nonumber \\ &\le& \mathbf{E}_{x\sim \gamma_n} [|g(\overline{x}_1,\ldots, \overline{x}_\ell)-g(\langle w_1,x \rangle, \ldots, \langle w_\ell, x \rangle)|] 
\end{eqnarray}
Now, observe that by definition, the term inside the expectation is uniformly bounded by $2$. On the other hand, if a point $x$ is good, then by (\ref{eq:good}) and exploiting $g$ is $t^{-1/2}$-Lipschitz, then $|g(\overline{x}_1,\ldots, \overline{x}_\ell)-g(\langle w_1,x \rangle, \ldots, \langle w_\ell, x \rangle)| \le \epsilon$. Since the fraction of good points is at least $1-\frac{\epsilon^2}{\ell}$, we get that for any $g \in \mathsf{Cover}(t, \ell, \delta)$, 
$$
\big| \mathbf{E}_{x\sim \gamma_n} [|P_tf(x) - g(\overline{x}_1,\ldots, \overline{x}_\ell)|] - \mathbf{E}_{x\sim \gamma_n} [|P_tf(x) - g(\langle w_1,x \rangle, \ldots, \langle w_\ell, x \rangle)|] \big| \le 2\epsilon. 
$$
Now a standard Chernoff bound implies that with for any $g \in \mathsf{Cover}(t, \ell, \delta)$, $\mathcal{O}_g$ is within $\pm \epsilon/2$ $\mathbf{E}[|P_tf(x) - g (\langle w_1, x\rangle, \ldots, \langle w_\ell, x \rangle)|]$ with probability $1- \frac{\epsilon}{10 \cdot |\mathsf{Cover}(t, \ell, \delta)|}$. Thus, by a union bound, with probability $1-\frac{\epsilon}{10}$, for all $g \in \mathsf{Cover}(t, \ell, \delta)$, $\mathcal{O}_g$ is within $\pm \epsilon/2$ of $\mathbf{E}[|P_tf(x) - g (\langle w_1, x\rangle, \ldots, \langle w_\ell, x \rangle)|]$. This finishes the proof.

\end{proof}
We are now ready to prove Theorem~\ref{thm:main-affine-invariant}. 
\begin{proofof}{Theorem~\ref{thm:main-affine-invariant}}
Set $t = \frac{\epsilon^4}{900 s^2}$ (this is the same setting as Lemma~\ref{lem:find-dirs} and Lemma~\ref{lem:test-hypothesis}). Observe that with this choice of $t$, since $f$ has surface area bounded by $s$, then by Proposition~\ref{prop:noise-stab-surf-1}, we get that 
$$
\mathbf{E}[|P_tf(x) - f(x)|] \le \sqrt{\mathbf{E}[|P_tf(x) - f(x)|^2]} \le \frac{\epsilon}{\sqrt{5}}. 
$$
We now run the algorithm \textsf{Find-candidate-directions} with noise parameter $t$, error parameter $\epsilon$ and surface area parameter $s$. We are guaranteed that with probability $1-\epsilon$, we will get $\ell \le k$ directions $y_1, \ldots, y_\ell$ which are $\gamma/2$-linearly independent and $P_tf$ is $\epsilon$-close to a junta on the subspace $\mathsf{span}(v_1, \ldots, v_\ell)$ where $v_i = DP_tf(y_i)$ (call this event $\mathcal{E}_1$). The query complexity of this (from Lemma~\ref{lem:find-dirs}) is $(s \cdot k /\epsilon)^{O(k)}$.

Next, we run the routine \textsf{Estimate-closest-hypothesis} with the directions $y_1, \ldots, y_\ell$, surface area parameter $s$, error parameter $\epsilon$. 
Observe that the query complexity of \textsf{Estimate-closest-hypothesis} is also  $(s \cdot k /\epsilon)^{O(k)}$. Thus, the total query complexity remains $(s \cdot k /\epsilon)^{O(k)}$. 

By guarantee of \textsf{Estimate-closest-hypothesis}, we have the following: there is an orthonormal basis $(w_1, \ldots, w_\ell)$ of $\mathsf{span}(DP_tf(y_1), \ldots, DP_tf(y_\ell))$ such that 
\[
\mathbf{E}[|P_tf(x) - g(\langle w_1,x\rangle, \ldots, \langle w_\ell,x\rangle)|] \le \min_{g^\ast \in \mathsf{Cover}(t,\ell,\delta)} \mathbf{E}[|P_tf(x) - g^\ast(\langle w_1,x\rangle, \ldots, \langle w_\ell,x\rangle)|] + 5\epsilon.
\]
However, conditioned on $\mathcal{E}_1$, $P_tf $ is $\epsilon$-close to a junta on $\mathsf{span}(DP_t f(y_1), \ldots, DP_t f(y_\ell))$. By Theorem~\ref{thm:net}, this  implies that the quantity $\min_{g^\ast \in \mathsf{Cover}(t,\ell,\delta)} \mathbf{E}[|P_tf(x) - g^\ast(\langle w_1,x\rangle, \ldots, \langle w_\ell,x\rangle)|] \le 3 \epsilon$. This means that if we output the function $g$, then 
$\mathbf{E}[|P_tf(x) - g(\langle w_1,x\rangle, \ldots, \langle w_\ell,x\rangle)|] = O(\epsilon)$. {Consider the subspace $V$ spanned by vectors $\{DP_tf(y)\}_{y \in \mathbb{R}^n}$. Note that $\mathsf{dim}(V) \le k$ and $V$ is a relevant subspace for $f$. Thus, $w_1, \ldots, w_\ell$ can be extended to a basis for $V$, finishing the proof. 
}


\end{proofof}

{
\begin{remark}~\label{rem:gaussian}
A crucial point about the routine \textsf{Find-invariant-structure}, which will be useful in the next section, is the following: The marginal distribution of all the queries is distributed as the standard $n$-dimensional Gaussian distribution $\gamma_n$. To see this, note that 
\begin{enumerate}
\item In the routine \textsf{Find-candidate-directions} , each of the directions $y_i$ is sampled from $\gamma_n$. Further, for $y_i$ and $y_j$  which are i.i.d. samples from $\gamma_n$, the queries made to the oracle for $f$ in computing $\langle DP_tf(y_i), DP_tf(y_j ) \rangle$ are also distributed as $\gamma_n$ (see Lemma~\ref{lem:inner-product-1}). 
\item  In the routine \textsf{Estimate-closest-hypothesis}, the points $x_i$ are sampled from $\gamma_n$ as are the directions $y_j$ (which are output of \textsf{Find-candidate-directions}). With this, the queries made to the oracle for $f$ for computing $f_{\partial, \xi, t, y_j}(x_i) $ are distributed as $\gamma_n$ (see Lemma~\ref{lem:compute-derivative-x}). 
\item One minor subtlety is that  while each sampled $y_j$ comes from $\gamma_n$, 
as stated, our algorithm \textsf{Find-invariant-structure} is adaptive.  Consequently, the above two items do not imply that the marginal distribution of all queries is coming from $\gamma_n$. 
The  cause of non-adaptivity is that in the routine \textsf{Find-candidate-directions}, while we sample each $y_j$ from $\gamma_n$, subsequently, we only use a subset of the sampled $y_j$'s (namely, the subset $S$). However, we can easily make this algorithm non-adaptive at no asymptotic increase in the sample complexity.  This is because the number of candidate directions sampled by the procedure \textsf{Find-candidate-directions} is at most  $k \cdot T_{\mathsf{succ}} = \mathsf{poly}(k \cdot s/\epsilon)$. We can run the subsequent routines 
namely \textsf{Compute-ortho-transform} and \textsf{Estimate-closest-hypothesis} with all the $y_j$'s instead of just those in set $S$ but only use those which are part of the set $S$ output by \textsf{Find-candidate-directions}. This will only increase the query complexity by a factor of $\mathsf{poly}(k \cdot s/\epsilon)$. 
\end{enumerate}
\end{remark}}



\newcommand{\normal}{\mathcal{N}}

\section{A lower bound in terms of surface area} \label{sec:lb}

The query complexity of our testing algorithm depends
on the surface area of the set being tested. In this section,
we prove that a polynomial dependence on surface area is
necessary for non-adaptive tester, by proving a lower bound for distinguishing
1-juntas and 2-juntas in two dimensions.
In particular, we show the following theorem. 
\begin{theorem}~\label{thm:lb}
Any non-adaptive algorithm which can distinguish between a $1$-junta with surface area at most $s$ versus $\Omega(1)$-far from a linear $1$-junta makes at least $s^{\frac{1}{10}}$ queries. 
\end{theorem}

To prove this theorem, as is standard, we will use the Yao's minimax lemma.  More specifically, we will describe a distribution $D_1$
over 1-juntas with surface area at most $\Theta(s)$ and a distribution
$D_2$ over functions that are far from 1-juntas and have surface area $\Theta(s)$, such that
for any choice of $x_1, \dots, x_n \in \R^2$ with $n = O(s^{1/10})$,
if $f \sim D_1$ and $g \sim D_2$ then $(f(x_1), \dots, f(x_n))$
and $(g(x_1), \dots, g(x_n))$ have almost the same distribution.

We begin with the description of $f \sim D_1$: let $\theta \in \R^2$
be a uniformly random unit vector. Choose $a_1, \dots, a_{s-1}$
uniformly from $[-1, 1]$, and then put them in increasing order.
We also set $a_0 = -1$ and $a_{s} = 1$.
Then choose independent
random bits $b_1, \dots, b_{s}$
and define $f$ by
\[
    f(x) = \begin{cases}
        b_i &\text{if $a_{i-1} < \langle x, \theta \rangle \le a_i$ for some $i \in \{1, \dots, s\}$} \\
        1 &\text{otherwise}.
    \end{cases}
\]
Clearly, such a function $f$ is a 1-junta,
and its surface area is at most $s+1$
because the boundary of $\{f = 1\}$ is a collection of at most $s+1$
lines, and each line has surface area at most $1/\sqrt{2\pi}$.

To describe the construction of $g \sim D_2$, we begin with
the same collection of random variables as before (i.e., $\theta$,
$a_1, \dots, a_{s-1}$, $b_1, \dots, b_s$).
Let $\theta^\perp$ be a $90^\circ$ clockwise rotation of $\theta$,
choose $z \in [-1, 1]$ independent of the other random variables,
and define $g$ by
\[
    g(x) = \begin{cases}
        b_i \sgn(\langle x, \theta^\perp \rangle - z) &\text{if $a_{i-1} < \langle x, \theta \rangle \le a_i$ for some $i \in \{1, \dots, s\}$} \\
        1 &\text{otherwise}.
    \end{cases}
\]
Note that the boundary of $\{g = 1\}$ is contained in at most $s+2$
lines, and so it has surface area at most $s+2$.
We will prove below that (with high probability) functions drawn from $D_2$ are far from 1-juntas.
Then the following Theorem will demonstrate that testing 1-juntas with surface area $\Theta(s)$
requires $\mathsf{poly}(1/s)$ queries.

\begin{theorem}\label{thm:surface-area-lower-bound}
    For any query set $x_1, \dots, x_n$ with $n \le s^{1/10}$, if $f \sim D_1$ and $g \sim D_2$
    then the distributions of $(f(x_1), \dots, f(x_n))$ and $(g(x_1), \dots, g(x_n))$ are $C s^{-1/10}$-close
    in total variation distance.
\end{theorem}

In order to study the distinguishability of $D_1$ and $D_2$, we give a slightly different
description of $f \sim D_1$ and $g \sim D_2$:
for $i = 1, \dots, s$ set
\begin{eqnarray*}
    S_i^+ &=&\{x: a_{i-1} < \langle x, \theta \rangle \le a_i \text{ and } \langle x, \theta^\perp \rangle \ge z\} \\
    S_i^- &=&\{x: a_{i-1} < \langle x, \theta \rangle \le a_i \text{ and } \langle x, \theta^\perp \rangle < z\} \\
    S_i &=& S_i^- \cup S_i^+,
\end{eqnarray*}
and note that $f$ was defined by independently assigning a random $\pm 1$ value on each set $S_i$,
while $g$ was defined by independently assigning opposite random $\pm 1$ values on each pair
$S_i^+$, $S_i^-$.
Also, $f$ and $g$ are both identically one on $\R^2 \setminus \bigcup S_i$.

Let $x_1, \dots, x_n$ be the set of query points, and consider the event
that for every $i$, at least one of $S_i^+$ or $S_i^-$ contains no point in $x_1, \dots, x_n$;
call this event $A$.
Then $A$ depends on $x_1, \dots, x_n$, $\theta$, and $a_1, \dots, a_{s-1}$, but not on $b_1,\dots,b_s$.
Thanks to the description of $f$ and $g$ above, conditioned on
$A$ the random variables $(f(x_1), \dots, f(x_n))$ and $(g(x_1), \dots, g(x_n))$
have the same distribution. In particular, we can couple $f$ and $g$ so that
$(f(x_1), \dots, f(x_n)) = (g(x_1), \dots, g(x_n))$ with probability at least $1 - \Pr[A]$,
and so we will prove Theorem~\ref{thm:surface-area-lower-bound} by showing that for
any choice of $x_1, \dots, x_n$ with $n \le s^{1/10}$,
$\Pr[A] \le C s^{-1/10}$.
To do this, we will divide the pairs $(x_i, x_j)$ into
``close'' pairs and ``far'' pairs: we say that $x_i$ and $x_j$ are $\delta$-close if
$|x_i - x_j| \le \delta$, and $\delta$-far otherwise.

The following lemma will complete the proof of Theorem~\ref{thm:surface-area-lower-bound}, because
it implies that with high probability no pair of points lies in the same strips $S_i$,
but on different sides of the line $\{x: \langle x, \theta^\perp \rangle = z\}$.

\begin{lemma}\label{lem:close-and-far}
    Suppose that $n \le s^{1/10}$ and set $\delta = s^{-1/3}$. For any set $x_1, \dots, x_n$, with probability
    at least $1 - C s^{-1/10}$:
    \begin{enumerate}
        \item every pair of points $x_i, x_j$ that are $\delta$-far do not belong to the same set $S_k$
            for any $k \in \{1, \dots, s\}$.
        \item every pair of points $x_i, x_j$ that are $\delta$-close
            lie on the same side of the line $\{x: \langle x, \theta^\perp\rangle = z\}$.
    \end{enumerate}
\end{lemma}

The first step of Lemma~\ref{lem:close-and-far} is the simple observation that far points remain
reasonably far even after projecting them in the direction $\theta$.
\begin{lemma}\label{lem:random-inner-product}
    For all sufficiently small $\delta$ and any $x \in \R^2$, $\Pr(|\langle \theta, x\rangle| \le \delta |x|) \le \delta$.
\end{lemma}

\begin{proof}
    If $\phi$ is the angle between $\theta$ and $x$ then $|\langle \theta, x\rangle| \le \delta |x|$
    exactly when $|\cos \phi| \le \delta$, which has probability $\frac{\cos^{-1}(-\delta) - \cos^{-1}(\delta)}{\pi}$.
    Since $\cos^{-1}$ has derivative 1 at zero, this is approximately $\frac{2}{\pi} \delta$ for small $\delta$.
    In particular, if $\delta > 0$ is sufficiently small then this probability is at most $\delta$.
\end{proof}

\begin{proof}[Proof of Lemma~\ref{lem:close-and-far}]
    Let $\ell_k = \ell_k(\theta, a_k)$ be the line $\{x: \langle x, \theta \rangle = a_k\}$.
    By Lemma~\ref{lem:random-inner-product} applied to $x_i - x_j$, if $x_i$ and $x_j$ are $\delta$-far
    then with probability at least $1 - \delta$, $|\langle \theta, x_i - x_j \rangle| \ge \delta^2$.
    By a union bound, with probability at least $1 - n^2 \delta$,
    $|\langle \theta, x_i - x_j \rangle| \ge \delta^2$ for every $\delta$-far pair $x_i, x_j$; from now
    on, we will condition on this event (call it $\Omega_1$) occurring.
    
    Now, if either $\langle \theta, x_i \rangle$
    or $\langle \theta, x_j \rangle$ lies outside of the interval $[-1, 1]$ then $x_i$ and $x_j$
    do not both lie in any single $S_k$. On the other hand, if both $\langle \theta, x_i \rangle$
    and $\langle \theta, x_j \rangle$ lie in $[-1, 1]$, then each line $\ell_k$ has (independently)
    probability $|\langle \theta, x_i - x_j \rangle| \ge \delta^2$ to ``split'' $x_i$ from $x_j$.
    Hence, with probability at least $1 - (1 - \delta^2)^{s-1} \ge 1 - \exp(\delta^2(s-1))$, there will
    be a line $\ell_k$ that splits $x_i$ from $x_j$, and so they will not belong to any single set $S_k$.
    Taking a union bound over all pairs $x_i, x_j$, we see that (conditioned on $\Omega_1$) with probability
    at least $1 - n^2 \exp(-\delta^2(s-1))$, no pair of $\delta$-far points lands in the same $S_k$.
    Removing the conditioning on $\Omega_1$ changes the probability bound to
    $1 - n^2 \delta - n^2 \exp(-\delta^2(s-1))$, which with our choice of parameters is at least
    $1 - C s^{-1/10}$.

    If $x_i$ and $x_j$ are $\delta$-close then $|\langle \theta^\perp, x_i - x_j\rangle| \le \delta$, and hence
    the probability that they land on opposite sides of the line
    $\{x: \langle x, \theta^\perp \rangle = z\}$ is at most $O(\delta)$. By a union bound over all pairs,
    with probability at least $1 - C n^2 \delta \ge 1 - C s^{-1/10}$,
    every pair of $\delta$-close $x_i, x_j$ land on the same side of that line.
\end{proof}

\subsection{$D_2$ is far from a 1-junta}

So far, we have shown that one cannot distinguish $D_1$ from $D_2$ from few samples. It remains to
show that functions
from $D_2$ are far (with high probability) from 1-juntas, it will follow that one cannot
1-juntas with $O(s)$ surface area with fewer than $s^{1/10}$ queries.

\begin{theorem}\label{thm:far-from-1-junta}
    There is a constant $c > 0$ such that with probability at least $1 - \mathsf{poly(1/s)}$
    over $g \sim D_2$, $g$ is $c$-far from every 1-junta.
\end{theorem}




Now recall that the construction of $D_1$ and $D_2$ involved dividing up the
strip $\{x: \langle \theta, x \rangle \in (-1, 1]\}$ into $s$ strips $S_1,
\dots, S_s$ and assigning random values on each strip.
Since both the construction of $D_2$ and the notion of distance to a 1-junta are rotationally
invariant, we will assume from now on that $\theta = e_1$, which means that the strips
$S_1, \dots, S_s$ are vertically oriented.
Let $U^+ = \bigcup_{i: b_i = 1} (a_i, a_{i+1}]$ and let $U^- = [-1, 1] \setminus U^+$.

\begin{definition}
    Let $I \subset [-1, 1]$ be an interval. We say that $I$ is \emph{$\delta$-balanced} if
    of both $|I \cap U^+|$ and $|I \cap U^-|$ are at least $\delta |I|$, where $|\cdot|$ denotes the one-dimensional
    Lebesgue measure.
    We say that $I$ is \emph{wide} if $|I| \ge \frac{1}{s}$.

    We extend these definitions to strips in two dimensions: say that $I \times \R$ is $\delta$-balanced (resp. wide)
    if $I$ is $\delta$-balanced (resp. wide).
\end{definition}

\begin{definition}
    For any line $\ell \subset \R^2$, we say that $\ell$ is $\delta$-balanced if both
    \[
        \int_{\ell \cap U^+} e^{-|x|^2/2} \, dx \quad \text{and} \quad
        \int_{\ell \cap U^-} e^{-|x|^2/2} \, dx
    \]
    are at least
    \[
        \delta \int_{\ell} e^{-|x|^2/2}\, dx.
    \]
\end{definition}

We will now describe the outline of Theorem~\ref{thm:far-from-1-junta}'s proof: note that if $h$ is a $1$-junta
then $h(x) = \tilde h(\langle \phi, x\rangle)$ for some $\phi$.
Now, Fubini's theorem implies that
\[
    2\pi \|h - g\|_1
    = \int_{\R^2} e^{-|x|^2/2} |h - g| \, dx
    = \int_\R \int_{\{x: \langle x, \phi^\perp \rangle = a\}} e^{-|x|^2/2} |\tilde h(a) - g(x)| \, dx \, da.
\]
Now, whenever the line $\{x: \langle x, \phi^\perp\}$ is $\delta$-balanced, the inner integral is
at least $\delta \int e^{-|x|^2 / 2} \, dx$. Therefore, in order to prove Theorem~\ref{thm:far-from-1-junta},
it suffices to show that there is a constant $\delta$ such that at least a constant fraction of the lines
$\{x: \langle x, \phi^\perp \rangle = a\}$ are $\delta$-balanced.
To be precise, let $L(\phi)$ be the set of lines of the form $\{x: \langle
\phi^\perp, x\rangle = a\}$ for $a \in [-10, 10]$. Since $e^{-|x|^2/2}$ is bounded from below on $[-10, 10]$,
it suffices to show that there is a constant $\delta > 0$ such that with high probability,
for every $\phi$, a constant fraction of $\ell \in L(\phi)$ are $\delta$-balanced. For the remainder of the
section, we will focus on proving the preceding statement.

We will consider two cases depending on $\phi$: if the lines in $L(\phi)$ are
``steep,'' then these lines will be balanced because a constant fraction of them will
cross the horizontal line $\{x: x_2 = z\}$ near the middle of a strip. Since the value of $g$ on a strip changes sign at 
that horizontal line,
this will imply that such a line is balanced.
On the other hand, if the lines are not steep, then they will be balanced because they cross many
strips, and $g$ will tend to take different values on different strips.

We will first deal with the case of steep lines. In this case, it is deterministically the case that
$g$ is far from $h$.

\begin{lemma}\label{lem:wide-strips}
    At least half of the points on the line segment from $(-1, 1/2)$ to $(1, 1/2)$ are
    in a wide strip $S_i$.
\end{lemma}

\begin{proof}
    There are $s$ strips in total, and so the narrow ones can take up at most a total width of 1,
    which is only half of the width of the line segment in question.
\end{proof}

\begin{lemma}
    There is a constant $c > 0$ such that if the absolute value of the slope of $\{x: \langle \phi^\perp, x \rangle = 0\}$
    is at least $s$ then a $c$-fraction of $\ell \in L(\phi)$ are $c$-balanced.
\end{lemma}

\begin{proof}
    We may assume without loss of generality that $z \le 0$.
    By Lemma~\ref{lem:wide-strips}, at least a constant fraction of $\ell \in L(\phi)$
    intersect the line $\{x: x_2 = 1/2\}$  in the middle third of a wide strip $S_k$. In this case,
    $\ell$ belongs to $S_k^+$ for a distance of at least 1/3, and to $S_k^-$ for a distance
    of at least 1/3, and it follows that $\ell$ is $c$-balanced for a constant $c$
    depending on the minimum and maximum values of $e^{-|x|^2/2}$ for $x \in [-1, 1]^2$.
\end{proof}

For the remainder of the section we will deal with lines that are not steep. For $k$
with $2^{-k} \le 2/s$, consider an interval of the form $[j 2^{-k}, (j+1) 2^{-k}] \subset [-1, 1]$;
let $D_k$ be the set of all such intervals.

\begin{lemma}\label{lem:balanced-dyadics}
    There is a constant $C$ such that
    with probability at least $1 - \mathsf{poly}(1/s)$,
    for every $k$ for which $2^{-k} \le 2/s$, at least a $\frac{1}{C}$-fraction of the intervals $I \in D_k$ are $\frac{1}{C}$-balanced.
\end{lemma}

\begin{proof}
    For technical convenience, we will consider a slightly different way of generating the strips $S_i$.
    Instead of dividing $[-1, 1] \times \R$ using exactly $s-1$ vertical lines, we will take a Poisson number
    (with mean $s-1$) of vertical lines. We will prove the claim for this modified model, with a probability
    estimate of at least $1 - \exp(-\Omega(\sqrt s))$, and since a Poisson random variable is equal to its mean
    with probability $\mathsf{poly}(1/s)$, the claim will also follow for the original model.

    Our first claim is that for $2^{-k} \le 1/\sqrt s$, each interval in $D_k$ has a constant
    probability of being $\Omega(1)$-balanced.
    First, consider the largest $k$ for which $2^{-k} \le 2/s$. In this case, the width of each $I \in D_k$
    is within a factor 2 of $s$ (we will call such an interval a
    \emph{primitive} interval.
    It is easy to verify that for each $I \in D_k$, there is a constant probability that $I$ will intersect
    exactly two strips, each taking up at least 1/3 of the width of $I$, and that these two strips
    will receive different labels $b_i$. Hence, there is a constant probability that $I$ is $1/3$-balanced.

    Now consider $k$ for which $2^{-k} \le 1 / \sqrt s$. Every $I \in D_k$ is made up of
    $\Theta(s 2^k)$ primitive intervals, each of which has a constant probability of being balanced.
    Moreover (thanks to our Poissonized model) the events that different primitive intervals are balanced
    are independent. By Chebyshev's inequality, there is a constant probability that at least a constant
    fraction of $I$'s primitive intervals are $1/3$-balanced, and so $I$ has a constant probability of being
    $\Omega(1)$-balanced.
    This proves our first claim (that for each $2^{-k} \le 1/\sqrt s$, each interval in $D_k$ has a constant
    probability of being $\Omega(1)$-balanced). Now, for each such $k$ there are at least $\Omega(\sqrt s)$
    such intervals, and so a Chernoff bound implies that with probability at least $1 - \exp(-\Omega(\sqrt s))$,
    at least a constant fraction of these intervals are balanced. Taking a union bound over $k$
    proves the claim whenever $2^{-k} \le 1/\sqrt s$.

    For smaller $k$, we claim that with high probability, \emph{every} $I \in D_k$ is balanced.
    Indeed, such $I \in D_k$ contain at least $\sqrt s$ primitive intervals, and so a Chernoff bound
    implies that with probability $1 - \exp(-\Omega(\sqrt s))$, at least a constant fraction of those
    primitive intervals are balanced, and so $I$ is balanced also. We can take a union bound over all $k$
    and all $I$.
\end{proof}

To complete the proof of Theorem~\ref{thm:far-from-1-junta}, it remains to show that with high probability,
every non-steep line is balanced.

\begin{lemma}
    There is a constant $c > 0$ such that if the absolute value of the slope of $\{x: \langle \phi^\perp, x \rangle = 0\}$
    is at most $s$ then with probability at least $1 - \mathsf{poly}(1/s)$, a $c$-fraction of $\ell \in L(\phi)$ are $c$-balanced.
\end{lemma}

\begin{proof}
    Choose $k$ so that the slope of all lines in $L(\phi)$ are between $2^{k-1}$ and $2^{k}$.
    Consider a rectangle of the form $Q = [j 2^{-k}, (j+1) 2^{-k}] \times [-2, -1]$, where
the interval $[j 2^{-k}, (j+1) 2^{-k}]$ is balanced. Since the slope of $\phi$
is at most $2^{k}$, if the line $\ell$ intersects the rectangle $Q$ then it
crosses the entire vertical strip $[j 2^{-k}, (j+1) 2^{-k}] \times \R$ within the horizontal strip $[-3, 0]$.
Since the interval $[j 2^{-k}, (j+1) 2^{-k}]$ is balanced, it follows that the line $\ell$ is also balanced.
(We're assuming here, without loss of generality, that $z \ge 0$).

Finally, it is easy to verify that if a constant fraction of the intervals $[j 2^{-k}, (j+1) 2^{-k}]$
are balanced then a constant fraction of $\ell \in L(\phi)$ intersect with some rectangle of the form above.
By Lemma~\ref{lem:balanced-dyadics}, this completes the proof.
\end{proof}

 
\bibliography{allrefs}
\bibliographystyle{alpha}

\appendix
\section{Small net for noise attenuated linear juntas}~\label{sec:noise-attenuated}
In this section, we are going to prove the following theorem which essentially shows the existence of a small cover for noise stable linear juntas. 
{To state this theorem, we will require one crucial fact about noise attenuated functions (due to Bakry and Ledoux~\cite{Bakry:94}) 
\begin{lemma}~\label{prop:gradient-bound}~
Let $f: \mathbb{R}^n \rightarrow [-1,1]$. Then, $P_t f$ is $C_t$-Lipschitz for $C_t = O(t^{-1/2})$. 
\end{lemma}
For the rest of this section, we are going to use $C_t$ to denote this quantity. 
}
We can now state the main theorem of this section. 
\begin{theorem}~\label{thm:net}
For any error parameter $\delta>0$, noise parameter $t>0$ and $k \in \mathbb{N}$, there is a set of functions $\mathsf{Cover}(t,k,\delta)$ (mapping $\mathbb{R}^k$ to $[-1,1]$) such that the following holds: 
\begin{enumerate}
\item Let $f: \mathbb{R}^n \rightarrow [-1,1]$ and $W$ be a $k$-dimensional space such that $P_t f$ is $\delta$-close to a $W$-junta. Further, $(w_1, \ldots, w_k)$ be any orthonormal basis of $W$. Then, $P_t f$ is $3\delta$-close to $h(\langle w_1, x \rangle, \ldots, \langle w_k, x \rangle)$ for some $h \in \mathsf{Cover}(t,k,\delta)$. 
\item Every function in $\mathsf{Cover}(t, k, \delta)$ is $2C_t$-Lipschitz.
\item $\log |\mathsf{Cover}(t, k, \delta)| \le \left(\frac{C \sqrt k \log^2(1/\delta)}{\delta \sqrt t}\right)^k$.
\end{enumerate}
\end{theorem}
The proof of this theorem relies on the following two lemmas. 
\begin{lemma}~\label{lem:net-1}
For any $L>0$, error parameter  $\delta>0$ and $k \in \mathbb{N}$, there is a set $\mathsf{Cover}_{k,L,\delta}$  consisting of functions mapping $\mathbb{R}^k \mapsto [-1,1]$ such that the following holds:
\begin{enumerate}
\item For every $g: \mathbb{R}^k \rightarrow [-1,1]$ which is $L$-Lipschitz, there is a function $h \in \mathsf{Cover}_{k,L,\delta}$ such that $\mathbf{E}[|g(x) - h(x)|] \leq \delta$. 
\item Every function in $\mathsf{Cover}_{k,L,\delta}$ is $2L$-Lipschitz.
\item $\log |\mathsf{Cover}_{k,L,\delta}| \le \left(\frac{C L \sqrt k \log^2(1/\delta)}{\delta}\right)^k$.
\end{enumerate}
\end{lemma}
\begin{proof}
    Let $\mathcal{B}  = \{x: \Vert x \Vert_2 \le \sqrt{k} \cdot \log
    (100/\delta)\}$. Let $\mathcal{A}$ be a maximal $\delta/(2L)$-packing of
    $\mathcal{B}$ (that is, a maximal subset of $\mathcal{B}$ such that any two
    distinct points in $\mathcal{A}$ are at least $\delta/(2L)$ apart.
    It is well-known (see, e.g.~\cite{LedouxTalagrand}) that $\mathcal{A}$ is a $\delta/L$-net
    of $\mathcal{B}$ and that $|\mathcal{A}| \le (C L \sqrt k \log (1/\delta)/\delta)^k$
    (the $\sqrt k \log (1/\delta)$ term comes from the diameter of $\mathcal{B}$.

    For $f: \mathbb{R}^n \rightarrow [-1,1]$, we now define $f_{\mathsf{int}}: \mathcal{A} \to [-1, 1]$ by
    simply rounding $f$ to the nearest integer multiple of $\delta/100$. To check the Lipschitz
    constant of $f_\mathsf{int}$, note that if $x, y \in \mathcal{A}$ then
    \[
        |f_{\mathsf{int}}(x) - f_{\mathsf{int}}(y)| \le |f(x) - f(y)| + \delta/50
        \le L \|x - y\| + \frac{L}{25} \|x - y\|,
    \]
    where the last inequality used the fact that $f$ is $L$-Lipschitz and that
    every pair of points in $\mathcal{A}$ is $\delta/(2L)$-separated.
    In particular, $f_{\mathsf{int}}$ is $2L$-Lipschitz.
    Let $\mathsf{Cover}'$ be the set of all functions $f_{\mathsf{int}}$ obtained
    in this way. Then the size of $\mathsf{Cover}'$ is at most $\exp((C L \sqrt k \log^2(1/\delta) \delta^{-1})^k)$,
    because there are at most $C/\delta$ choices for the value of each point, and there are $|\mathcal{A}|$ points.
    Finally, we construct $\mathsf{Cover}_{k,L,\delta}$ by extending each function in $\mathsf{Cover}'$
    to a function $\R^n \to [-1, 1]$. McShane's Lemma~\cite{mcshane34} implies that this extension can be done
    without increasing its Lipschitz constant. Hence, properties 2 and 3 hold.

    To check property 1, note that if $x \in \mathcal{B}$ and $y \in \mathcal{A}$ is the closest point to $x$
    then
    \[
        |f(x) - f_{\mathsf{int}}(x)|
        \le |f(x) - f(y)| + |f(y) - f_{\mathsf{int}}(y)| + |f_{\mathsf{int}}(y) - f_{\mathsf{int}}(x)|
        \le 3L \|x - y\| + \delta/100 \le 4\delta.
    \]
    It then follows that
    \[
    \mathbf{E}[|f(x) - f_{\mathsf{int}}(x)|] \le 2 \cdot \Pr[x \not \in \mathcal{B}] + \max_{x \in \mathcal{B}} [|f(x) - f_{\mathsf{int}}(x)|] \le \delta + 4 \delta \le 5 \delta.
    \]
The last inequality just follows from the fact that a $k$-dimensional standard Gaussian is in a ball of radius $\sqrt{k} \log (1/\delta)$ with probability $1-\delta/2$. This proves property 1 modulo the constant $5$, which can be dropped by redefining $\delta$.
\end{proof}

\begin{lemma}~\label{lem:Lip-1}
Let $f: \mathbb{R}^n \rightarrow [-1,1]$ be a $C$-Lipschitz function. Further, for $\kappa>0$, let $g: \mathbb{R}^n \rightarrow [-1,1]$ be a $W$-junta such that $f$ is $\kappa$-close to $g$. Then, there is a function $f_W:\mathbb{R}^n \rightarrow [-1,1]$ which is $C$-Lipschitz and $W$-junta which is $2\kappa$-close to $f$. 
\end{lemma}
\begin{proof}
Reorient the axes so that $W$ is the space spanned by the first $\ell$-axes. Let us define the $W$-junta $f_{W}: \mathbb{R}^n \rightarrow [-1,1]$ defined as 
\[
f_W(x) = \mathbf{E}_{y_{\ell+1}, \ldots, y_n} [f(x_1, \ldots, x_\ell, y_{\ell+1}, \ldots, y_n)
\]
For any fixed choice of $x_1, \ldots, x_\ell$, we have
\[
\mathbf{E}_{x_{\ell+1}, \ldots, x_n}[|f(x) - f_W(x)|] \leq \mathbf{E}_{x_{\ell+1}, \ldots, x_n}[|f(x) - g(x)|]  +|g(x) - f_W(x)|. 
\]
However, the second term can be bounded as 
\[
|g(x) - f_W(x)| = \big| g(x) - \mathbf{E}{x_{\ell+1}, \ldots, x_n}[f(x_1, \ldots,x_\ell, x_{\ell+1} , \ldots, x_n)] \big|  \le \mathbf{E}{x_{\ell+1}, \ldots, x_n} \big[ \big| g(x) - f(x) \big|\big]
\]
The last inequality is simply Jensen's inequality. Combining these two, we get 
\begin{equation}~\label{eq:junta-diff}
\mathbf{E}_{x_{\ell+1}, \ldots, x_n}[|f(x) - f_W(x)|] \leq2 \cdot \mathbf{E}_{x_{\ell+1}, \ldots, x_n}[|f(x) - g(x)|]. 
\end{equation}
This in turn implies that 
\begin{equation}~\label{eq:junta-diff-1}
\mathbf{E}_{x_{1}, \ldots, x_n}[|f(x) - f_W(x)|] \leq2 \cdot \mathbf{E}_{x_{1}, \ldots, x_n}[|f(x) - g(x)|] \leq 2\cdot \kappa. 
\end{equation}
Finally, we see that 
\begin{eqnarray*}
|f_W(x) - f_W(y)|  &=& \big| \mathbf{E}_{x_{\ell+1}, \ldots, x_n}[f(x_1, \ldots, x_\ell, x_{\ell+1} , \ldots, x_n) - f(y_1, \ldots, y_\ell, x_{\ell+1},\ldots, x_n)] \\
&\leq&  \mathbf{E}_{x_{\ell+1}, \ldots, x_n} \big[ \big| f(x_1, \ldots, x_\ell, x_{\ell+1} , \ldots, x_n) - f(y_1, \ldots, y_\ell, x_{\ell+1},\ldots, x_n) \big| \big] \\
&\le&  \mathbf{E}_{x_{\ell+1}, \ldots, x_n} [ C \cdot \Vert  (x_1, \ldots, x_\ell) - (y_1, \ldots, y_\ell) \Vert_2] \le C \Vert x -y\Vert_2. 
\end{eqnarray*}
This finishes the proof. 
\end{proof}
With these two lemmas, we can now finish the proof of Theorem~\ref{thm:net}. 
{\begin{proofof}{Theorem~\ref{thm:net}}
First, we apply Lemma~\ref{prop:gradient-bound} to obtain that $P_t f$ is $C_t=O(t^{-1/2})$-Lipschitz. Since $P_t f$ is $\delta$-close to a $W$-junta, we obtain that $P_t f$ is $2\delta$ close to a $W$-junta $g$ which is $C_t$-Lipschitz (follows from Lemma~\ref{lem:Lip-1}). Let $ \mathsf{Cover}(t,k,\delta)=\mathsf{Cover}_{k,C_t,\frac{\delta}{2}}$ (constructed  in Lemma~\ref{lem:net-1}). By a rotation of the coordinates, it follows from the definition of $ \mathsf{Cover}(t,k,\delta)$ that there exists $h \in\mathsf{Cover}(t,k,\delta)$ such that $h( \langle w_1, x \rangle, \ldots, \langle w_k, x \rangle)$ is $\frac{\delta}{4}$ close to $g$. The required properties now  follow from Lemma~\ref{lem:net-1}. 
\end{proofof}}

\section{Some useful results from linear algebra}
The next lemma states for any $v_1, \ldots, v_\ell$ which are $(\eta,\gamma)$-linearly independent, we can 
find a set of vectors $(w_1, \ldots, w_\ell)$ (expressed as linear combination of $(w_1, \ldots, w_\ell)$) which is close to being an orthonormal basis of the $\mathsf{span}(v_1, \ldots, v_\ell)$ provided we have sufficiently good approximations of $\{\langle v_i, v_j \rangle\}_{1 \le i,j \le \ell}$. Now, modulo the \emph{quantitative estimates}, this is essentially just a consequence of a procedure such as the Gram-Schmidt orthogonalization. However, the complexity of our testing algorithm is dependent on the quantitative estimates, so we work out the linear algebra here. 

\begin{lemma}~\label{prop:linear}
Let $v_1, \ldots, v_\ell$ be a $(\eta, \gamma)$-linearly independent vectors. Then, for any error parameter $\nu>0$ and $\lambda =\lambda(\ell,\nu,\eta, \gamma)$ defined as 
$$
\lambda= 2 \frac{\nu}{\ell^2 \cdot \eta} \cdot \big( \frac{\gamma}{2\cdot \ell \cdot \eta}\big)^{3\ell+3}, 
$$ given numbers $\{\beta_{i,j}\}_{1\le i,j \le \ell}$ such that $|\beta_{i,j} - \langle v_i, v_j \rangle| \le \lambda$, we can compute numbers $\{\alpha_{i,j}\}_{1\le i,j \le \ell}$ such that: 
\begin{enumerate}
\item For $\xi (\ell, \eta, \gamma)$ defined as 
\[
\xi ( \ell, \eta, \gamma) =\sqrt{2\ell} \cdot \bigg(\frac{2 \ell \cdot \eta}{\gamma} \bigg)^{\ell+1}, 
\]
 we have 
$|\alpha_{i,j}| \le \xi ( \ell,\eta, \gamma)$.
\item There is an orthonormal basis $(w_1, \ldots, w_\ell)$ of $\mathsf{span}(v_1, \ldots, v_\ell)$ such that for $\Vert w_{i} - \sum_{j}\alpha_{i,j} v_j \Vert_2 \le \nu$. 
\end{enumerate}
\end{lemma}
\begin{proof}
 Consider the symmetric matrix $\Sigma\in\mathbb{R}^{\ell \times \ell}$ defined as $\Sigma_{i,j} = \langle v_i, v_j \rangle$. By Proposition~\ref{prop:sing-1}, $\Sigma$ is non-singular. Define the matrix $\Gamma = \Sigma^{-1/2}$. It is easy to see that the columns of 
 $V \cdot \Sigma^{-1/2}$ form an orthonormal basis of $\mathsf{span}(v_1, \ldots, v_\ell)$. Here $V = [v_1 | \ldots  | v_\ell]$. Of course, we cannot compute the matrix $\Sigma$ exactly and consequently, we cannot compute the matrix $\Sigma^{-1/2}$ either. Instead, if we define the matrix $\widetilde{\Sigma}$ as $\widetilde{\Sigma}(i,j) = \beta_{i,j}$, then observe that $\widetilde{\Sigma}$ is symmetric. Next, observe that Proposition~\ref{prop:sing-1}, we have that 
 $$
 \sigma_{\min}(\Sigma) = \sigma_{\min}^2(V) \geq \bigg(\frac{\gamma}{2 \cdot \ell \cdot \eta} \bigg)^{2 \ell +2}. 
 $$
 Define a  parameter $\rho$ as
 $$
 \rho = \frac{2 \nu}{\ell \cdot \eta} \cdot \bigg( \frac{\gamma}{2 \cdot \ell \cdot \eta}\bigg)^{\ell+1}. 
 $$
 Now, with this setting, observe that 
 $$
 \ell \cdot \lambda = \rho \cdot \bigg( \frac{\gamma}{2 \cdot \ell \cdot \eta} \bigg)^{2\ell +2} \le \rho \cdot \sigma_{\min} (\Sigma).  
 $$
Further, since entrywise, $\Sigma$ and $\widetilde{\Sigma}$ differ by at most $\lambda$, hence $\Vert \widetilde{\Sigma} - \Sigma \Vert_F\le \ell \cdot \lambda$.  First, by Weyl's inequality (Lemma~\ref{lem:Weyl}), we have that 
\begin{equation}~\label{eq:sigma-min}
\sigma_{\min}(\widetilde{\Sigma}) \ge  \sigma_{\min}({\Sigma})- \Vert \Sigma-\widetilde{\Sigma} \Vert_F  \ge (1- \rho) \cdot \sigma_{\min}(\Sigma).
\end{equation}
Thus, $\widetilde{\Sigma}$ is also psd. 
Now, we apply the matrix perturbation bound to matrices $\Sigma$ and $\widetilde{\Sigma}$ (Corollary~\ref{corr:mat-perturb} with parameter $c = \big( \frac{\gamma}{2\cdot \ell \cdot \eta}\big)^{2\ell+2}$)
 to obtain that 
 \[
 \Vert \Sigma^{-1/2} - \widetilde{\Sigma}^{-1/2} \Vert \leq \frac{\rho}{2 \big( \frac{\gamma}{2\cdot \ell \cdot \eta}\big)^{\ell+1}} = \frac{2\nu}{\ell \cdot \eta}. 
 \]
We now define $\alpha_{i,j} = \widetilde{\Sigma}^{-\frac12}(j,i)$. We also define $\beta_{i,j} = {\Sigma}^{-\frac12}(i,j)$. Note that the vectors $w_i = \sum_{i} \beta_{j,i} v_j$ forms an orthonormal basis. As the matrices $\Sigma^{-\frac12}$ and $\widetilde{\Sigma}^{-\frac12}$ are $\frac{2\nu}{\ell \cdot \eta}$ close in operator norm, this immediately implies item 2. To get item 1, we recall the following basic inequality for Frobenius norm of an inverse matrix. In particular, for a symmetric matrix $A \in \mathbb{R}^{\ell \times \ell}$, 
$\sigma_{\min}(A) \cdot \Vert A^{-1} \Vert_F \le \sqrt{\ell}$. Thus, 
\[
\Vert \widetilde{\Sigma}^{-1/2} \Vert_F  \le \frac{\sqrt{\ell}}{\sigma_{\min}(\widetilde{\Sigma}^{1/2})} = \sqrt{\frac{\ell}{\sigma_{\min}(\widetilde{\Sigma})}} \le
 \sqrt{2\ell} \cdot \bigg(\frac{2 \ell \cdot \eta}{\gamma} \bigg)^{\ell+1}. 
\]
The last inequality uses (\ref{eq:sigma-min}) and the fact that $\rho \le \frac12$. This immediately implies 
the first item. 
 
\end{proof}
\begin{proposition}~\label{prop:sing-1}
Let $v_1, \ldots, v_\ell$ be a $(\eta, \gamma)$-linearly independent vectors.  Let $V = [v_1 | \ldots | v_\ell]$. Then, the smallest singular value of $V$
is at least $(\frac{ \gamma}{2 \cdot \ell \cdot \eta})^{\ell+1}$. 
\end{proposition}
\begin{proof}
Let us set a parameter $\rho - \frac{\gamma}{2\ell \eta}$. Recall that if $\sigma_{\min}(V)$ is the smallest singular value of $V$, then 
\[
\sigma_{\min}(V) = \inf_{x : \Vert x \Vert_2=1} \Vert V \cdot x \Vert_2
\]
Let us try to lower bound the right hand side. To do this, let $x \in \mathbb{R}^n$ be any unit vector and note that  $V \cdot x = \sum_{1 \le i \le \ell} v_i \cdot x_i$. 
Now, let $j$ be the largest coordinate such that $|x_j| \ge \rho^{j}$ (note that there has to be such a $j$ since $x$ is a unit vector and $\rho<1/2$). Define $w = \sum_{i \le j} v_i x_i$. Then, observe that its component in the direction orthogonal to the span of $\{v_1, \ldots, v_{j-1}\}$ is at least $\gamma \cdot \rho^j$ in magnitude. On the other hand, $\Vert \sum_{i > j} v_i x_i \Vert_2 \le \rho^{j+1} \cdot \ell \cdot \eta$. By triangle inequality, we obtain that 
\[
\Vert \sum_{i} v_i x_i \Vert_2 \ge \Vert \sum_{i \le j} v_i x_i \Vert_2 - \Vert \sum_{i >  j} v_i x_i \Vert_2 \ge \gamma \cdot \rho^j - \ell \cdot \eta \cdot \rho^{j+1}  \ge  \frac{\gamma \cdot \rho^j}{2}. 
\]
The last inequality uses the value of $\rho$. This finishes the proof.
\end{proof} 

\

\end{document}